\documentclass{article}

\usepackage{soul} 
\usepackage[a4paper, total={6in, 9in}]{geometry}

\usepackage{graphicx} 
\usepackage{amsmath}
\usepackage{amssymb}
\usepackage{amsthm}
\usepackage{comment}
\usepackage{physics}
\usepackage{dsfont}
\usepackage{mathtools}
\usepackage[colorinlistoftodos]{todonotes}
\usepackage{cases}
\usepackage{titlesec}

\usepackage{mathtools}

\usepackage{appendix}

\usepackage{color}

\usepackage{cases}
\usepackage{empheq}
\usepackage{multirow}
\usepackage{arydshln}
\usepackage{booktabs}
\usepackage{authblk}
\usepackage{cite}

\definecolor{darkpastelgreen}{rgb}{0.01, 0.75, 0.24}

\usepackage[pdfpagelabels]{hyperref}
\hypersetup{
	breaklinks=true,   
	colorlinks=true,   
    citecolor=darkpastelgreen,
	pdfusetitle=true,  
}
\usepackage{cleveref}

\newcommand{\F}{\mathcal{F}}

\newcommand{\W}{\mathcal{W}}

\newcommand{\intinf}{\int_{-\infty}^\infty}

\newcommand{\intoinf}{\int_{0}^\infty}

\newcommand{\FTChiSq}{{|\widehat{\chi}(\omega)|}^2}

\newcommand{\FTChiSqlambda}{{|\widehat{\chi}_\lambda(\omega)|}^2}
\newcommand{\ee}{\mathrm{e}}
\newcommand{\ii}{\mathrm{i}}

\renewcommand{\dd}{\mathrm{d}}

\DeclareMathOperator{\sgn}{\text{sgn}}
\DeclareMathOperator{\sinc}{\text{sinc}}
\DeclareMathOperator{\Si}{Si}

\newtheorem{proposition}{Proposition}[section]
\newtheorem{definition}[proposition]{Definition}

\numberwithin{equation}{section}
\title{Waiting around for Unruh}

\author[1]{Leo J. A. Parry\thanks{leo.parry@nottingham.ac.uk}}
\author[2]{Diego Vidal-Cruzprieto\thanks{vidaleando@proton.me}}
\author[2,3]{Christopher J. Fewster\thanks{chris.fewster@york.ac.uk}}
\author[1]{Jorma Louko\thanks{jorma.louko@nottingham.ac.uk}}

\affil[1]{School of Mathematical Sciences,
University of Nottingham, 
\break Nottingham NG7 2RD, UK}
\affil[2]{Department of Mathematics, Ian Wand Building, Deramore Lane,
\break University of York, York YO10 5GH, UK}
\affil[3]{York Centre for Quantum Technologies, University of York, York YO10 5DD, UK}

\date{August 2025; revised November 2025.\\[2ex]
{\small Published in Class. Quantum Grav. \textbf{42}, 245012 (2025).}}

\begin{document}
\maketitle

\begin{abstract}
How long does a uniformly rotating observer need to interact with a quantum field in order to register an approximately thermal response due to the circular motion Unruh effect? 
We address this question for a massless scalar field in $2+1$ dimensions, defining the effective temperature via the ratio of excitation and de-excitation rates of an Unruh-DeWitt detector in the long interaction time limit. In this system, the effective temperature is known to be significantly smaller than the linear motion Unruh effect prediction when the detector's energy gap is small: the effective temperature tends to zero in the small gap limit, linearly in the gap. 
We show that a positive small gap temperature at long interaction times can be regained via a controlled long-time-small-gap double limit, provided the detector's coupling to the field is allowed to change sign. The resulting small gap temperature depends on the parameters of the circular motion but not on the details of the detector's switching. 
The results broaden the energy range for pursuing an experimental verification of the circular motion Unruh effect in analogue spacetime experiments. 
As a mathematical tool, we provide a new implementation of the long interaction time limit that controls in a precise way the asymptotics of both the switching function and its Fourier transform. 

\end{abstract}

\section{Introduction} \label{sec: Introduction}

The Unruh effect \cite{Fulling:1972md,Davies:1974th,Unruh1976,Fulling:2014} is the prediction in quantum field theory that a uniformly linearly accelerated detector with proper acceleration $a$ responds to the Minkowski vacuum of a relativistic quantum field as if the vacuum were a thermal state at the Unruh temperature
\begin{equation} 
\label{eq: Unruh temp}
    T_U = \frac{\hbar a}{2\pi c k_B}.
\end{equation}
The geometric observation behind this phenomenon is that the restriction of the Minkowski vacuum to a quadrant of Minkowski space, adapted to the boost Killing vector that generates the linearly accelerated trajectory, is a genuinely thermal state with respect to the detector's proper time, at the temperature \eqref{eq: Unruh temp} \cite{Unruh1976,Fulling:2014,Bisognano:1975ih,Bisognano:1976za,SEWELL198023,Sewell:1982zz}.  
As a temperature of $1\,\mathrm{K}$ requires an acceleration of approximately $2.4 \times 10^{20}\mathrm{ms^{-2}}$, an observational verification of this phenomenon currently lies well beyond experimental capabilities. Despite this challenge, observation of the Unruh effect would be significant, partly due to its deep connections with other phenomena, such as the Hawking effect~\cite{Hawking:1975vcx}, and the cosmological particle production \cite{Parker:1969au}, which may have played a critical role in the formation of the present-day structure of the Universe \cite{Mukhanov:2007zz}. For a comprehensive review of the Unruh effect and its applications, see \cite{Crispino:2007eb}.

Better prospects for observing the Unruh effect are provided by analogue spacetime systems~\cite{Unruh1981,Liberati,HeliumUniverse}, which simulate relativistic quantum fields in nonrelativistic hydrodynamical or condensed matter table-top experiments. In such systems, excitations in some observable property of a fluid, like the density, velocity potential, or fluid height, play the role of the quantum field, and crucially, the speed of light is replaced by the speed of sound of the fluid, thereby raising the Unruh temperature \eqref{eq: Unruh temp} by several orders of magnitude. Confirmation of stimulated Hawking emission~\cite{Weinfurtner:2010nu}, the classical mode conversion underlying the Unruh effect \cite{Leonhardt:2017lwm}, 
and (phononic) Hawking radiation and entanglement \cite{Steinhauer:2015saa} have been reported in such analogue systems.

Another challenge with observing the Unruh effect, even in analogue systems, is to be able to maintain a linearly accelerating system within the experimental apparatus for a length of time sufficient to resolve the phenomenon. This challenge can be overcome by considering circular \cite{BellLeinaas, Bell:1986ir, Unruh:1998gq, Leinaas:1998tu} rather than linear acceleration, for which an experiment can be conducted for an arbitrarily long time. A~further advantage of circular acceleration is that the lack of a condensed matter relativistic time dilation can be accounted for at the data analysis stage \cite{Retzker:2007vql, Gooding:2020scc, Biermann:2020bjh, Unruh:2022gso, Gooding:2025tfp}.
The circular motion Unruh effect has previously been considered in the context of electron beam depolarisation in accelerator storage rings \cite{BellLeinaas,Bell:1986ir,Costa:1994yx,Guimaraes:1998jf,Leinaas:1998tu, Unruh:1998gq}. For other work on the theory of the circular motion Unruh effect, see \cite{Denardo:1978dj,Letaw:1979wy, LetawStationary, Takagi:1986kn, Davies:1996ks, Korsbakken:2004bv, Chowdhury:2019set, Lochan:2019osm, Good_2020, Bunney:2023vyj, Bunney:2023ipk, Bunney:2024qrh}.
 
For uniform linear acceleration, the Unruh effect can be understood in terms of a Bogoliubov transformation that expresses the Minkowski vacuum in terms of the Rindler vacuum (the vacuum adapted to the boost Killing vector), together with the corresponding Rindler excitations~\cite{Unruh1976}. 
The spectrum of these excitations is thermal at the Unruh temperature~\eqref{eq: Unruh temp}, where $a$ is the proper acceleration of the trajectory. 
For uniform circular motion, by contrast, a Bogoliubov transformation that would express the Minkowski vacuum in terms of a `co-rotating' vacuum and `co-rotating' excitations thereon is not available \cite{Unruh1976, Denardo:1978dj, Letaw:1979wy, Davies:1996ks, Earman:2011zz}. 
Instead, an \textit{effective\/} temperature parameter can be defined operationally in terms of the ratio of excitations and de-excitations of a ``particle detector", which is typically modelled by a localised quantum system with discrete energy levels or by a localised mode of a quantum field \cite{Unruh1976,DeWitt1979}. Although the effective temperature in circular motion depends not just on the proper acceleration but on all the parameters of the orbit and on the internal energy spacing of the detector, the effective temperature largely agrees with the Unruh temperature \eqref{eq: Unruh temp} over most of the parameter space \cite{Biermann:2020bjh, Unruh:1998gq, Good_2020}. It is in this sense that the response of such a detector is \textit{approximately\/} thermal and the effective temperature therefore provides a useful quantifier of the circular motion Unruh effect.

In the recent analogue spacetime proposals to observe the circular motion Unruh effect in 
Bose-Einstein condensates~\cite{Gooding:2020scc,Gooding:2025tfp} and 
superfluid Helium thin-films~\cite{Bunney:2023ude}, 
the quantum field that is simulated is a $(2+1)$-dimensional massless scalar field. In this system, the circular acceleration effective Unruh temperature is of the same order of magnitude as the linear acceleration prediction over most of the parameter space, but it is much smaller when the internal energy spacing of the accelerating detector is small and the interaction with the quantum field lasts for a long time~\cite{Biermann:2020bjh}. 
This phenomenon is characteristic of circular motion in $2+1$ dimensions for a massless field, and its mathematical origin is in the weak fall-off of the field's Wightman function along the circular worldline~\cite{Parry:2024jrm}. The phenomenon hence presents a challenge for experiments that would operate in the small energy spacing regime. 
We shall show that the challenge can be alleviated by allowing the coupling between the ambient field and the probe to change sign in a suitably controlled way.

Mathematically, the purpose of this paper is to investigate to what extent the assumption of an infinitely long interaction duration is responsible for the disparity between the circular motion effective temperature and the Unruh temperature in the small gap regime. In brief: \textit{How long does an observer undergoing circular motion in\/ $2+1$ dimensions need to wait to register approximate thermality when the observer's internal energy spacing is small?} Physically, the answer has clear resource implications for any experiment.

A similar question was investigated in \cite{Fewster:2016ewy} for a uniformly linearly accelerated two-level detector with a large energy gap. 
In that setting, the waiting time for approximate thermality was found to grow with the energy gap, at a rate that depends on how the large interaction time is modelled. 

The model we employ here is as follows. We consider a massless scalar field in $2+1$ dimensional Minkowski spacetime, prepared initially in the Minkowski vacuum. We probe the field with an Unruh-DeWitt (UDW) detector, a pointlike two-level quantum system coupled linearly to the quantum field \cite{Unruh1976,DeWitt1979}, and we take the interaction to be sufficiently weak that first order perturbation theory in the coupling strength is valid. In $3+1$ dimensions, this model captures the essence of the interaction between an atom and the electromagnetic field when angular momentum interchange is negligible \cite{Martin-Martinez:2012ysv,Alhambra:2013uja}. 
To control the interaction duration, we assume that the  
interaction strength is time dependent and proportional to a  switching function that depends on the detector's proper time, and the switching function either has compact support or is appropriately suppressed in the distant past and future. The profile of the switching then determines the interaction duration, or the effective interaction duration when the support is not compact. 
The long interaction limit is implemented by stretching the profile in a controlled way. 

The switching profile could be stretched in various ways, two of which were considered in~\cite{Fewster:2016ewy}: 
adiabatic scaling, in which the entire switching function is scaled by a constant in its argument, and plateau scaling, in which only the duration of an intermediate plateau interval is scaled, leaving the finite duration switch-on and switch-off intervals unscaled. However, there are other possibilities, and to provide a broad framework
we introduce a new implementation of the long interaction time limit, which we call an Asymptotically Scaled Switching Family (ASSF). This is a family of functions $\chi_\lambda$ ($\lambda>0$) obeying conditions set out in Definition~\ref{def:assf}, where the long time limit corresponds to $\lambda\to\infty$. The definition
encompasses as special cases both the adiabatic and plateau scalings of \cite{Fewster:2016ewy} but is significantly more general, controlling in a precise way the asymptotics of both the switching function and its Fourier transform. Some of our results are developed for ASSFs in general, while others are given for adiabatic and plateau scalings, making use of ASSF properties to simplify the technical proofs by placing them in the appropriately general and rigorous setting. However, we anticipate that many of the latter results are generalisable to broader classes of ASSFs.

We first consider ASSFs in which the functions $\chi_\lambda$ do not change their sign and are bounded uniformly in~$\lambda$. It is shown that the effective temperature vanishes in a successive limit of $\lambda\to\infty$ and $E\to 0$, taken in either order. We then specialise to adiabatically scaled and plateau scaled switching functions, maintaining the assumption that the switching function has a uniform sign, but considering a simultaneous limit of $\lambda\to \infty$ and $E\to 0$, assuming that these parameters are related in some way -- we choose an inverse power law. We show that 
the effective temperature still tends to zero in the small gap limit, although the rate at which the temperature approaches zero can be slowed to an arbitrarily small power of the gap by taking the interaction duration to increase relatively slowly when the detector gap decreases. 

We then consider switching functions that do change their sign. For adiabatic scaling, in the simultaneous long-time and small-gap limit described above, we show that a necessary condition for the effective temperature to have a positive small gap limit is that the time-average of the switching function vanishes. We present explicit examples of adiabatically scaled zero-time-average switching functions for which a positive small gap temperature is indeed attained. 
We further show that any adiabatically scaled, non-compactly supported switching function with this small-gap limit property
can be modified, subject to mild falloff conditions, into a compactly supported switching function that has the same limit property under a scaling that is asymptotically adiabatic. Finally, we show that similar constructions are not available with plateau scaling. 

Crucially, under the switching constructions described above, the recovered small-gap effective temperature depends only on the parameters of the circular motion and is independent of the details of the switching. These switching constructions therefore reveal properties of the worldline and the quantum field theory, rather than properties of a deftly engineered switching. 

Switching functions with sign changes appear not to have received attention in the literature. However, there is nothing that prevents us \textit{a priori\/} from considering such switching functions, or from constructing them in laboratory experiments. 
In fact, they arise naturally in entanglement harvesting protocols in which causally disconnected local laser pulses become entangled through their interaction with, for example, an electromagnetic field \cite{Lindel:2023rfi} or a Bose-Einstein condensate~\cite{Gooding:2023xxl}. In this setting, the switching function is given by the coherent amplitude of the electric field of the laser, which can be given any profile, and in particular, it can be engineered to change sign. In the Bose-Einstein condensate experiment proposed in~\cite{Gooding:2020scc, Gooding:2025tfp}, the engineering would be related to the detuning of the probe laser beam sidebands around a resonance frequency of the condensate.

The paper is structured as follows. In Section \ref{sec: Preliminaries} we introduce the preliminary theory: the UDW detector model and its response in stationary motion, approximate thermality and the effective temperature, 
and the long time limit. 
Section \ref{sec:assf} introduces ASSFs as a new implementation of the long interaction time limit, proving properties that will be used in the later sections. 
Section \ref{sec:circmotion}
specialises to $2+1$ circular motion and considers the long-time and small-gap limits when the limits are taken in succession in either order. 
Section \ref{sec: adiab / plat scalings} is an interlude that establishes the notation for the adiabatic and plateau scalings. 
In Section~\ref{sec: Small gap long time}, we compute the simultaneous long time with small gap asymptotics with adiabatic and plateau scaled switching functions, under the assumption of no sign changes. 
Switchings that change sign are considered in Section~\ref{sec: negative switchings}. 
Section \ref{sec:conclusions} presents a summary and conclusions. 
Proofs of technical results are delegated to five appendices. 

We use units in which $\hbar=k_B = c =1$. 
We work in $(2+1)$-dimensional Minkowski spacetime 
with standard Minkowski coordinates 
$(t,x,y) = (t, \mathbf{x})$, in which the Minkowski metric $\eta$ reads 
$\dd s^2 = - \dd t^2 + \dd x^2+\dd y^2$. 
Spacetime points are denoted by sans serif letters. 
In asymptotic formulae, 
$f(x)=O(x)$ denotes that $f(x)/x$ is bounded in the limit of interest, and $f(x) = o(x)$ denotes that $f(x)/x\to0$ in the limit of interest. The Heaviside theta function $\Theta(x)$ is defined as
\begin{equation}
    \Theta(x) = 
    \begin{cases}
        1 & \text{ for } x\geq0 \\
        0 & \text{ for } x<0,
    \end{cases}
\end{equation}
and the signum function $\sgn(x)$ is defined as
\begin{align}
    \sgn(x) = 
    \begin{cases}
        1 & \text{ for } x>0 \\
         0 & \text{ for } x=0\\
        -1 & \text{ for } x<0.
    \end{cases}
\end{align}
The Fourier transform is defined as
\begin{equation}
    \widehat{f}(\omega) = \intinf \dd t \, \ee^{-\ii\omega t} f(t).
    \label{eq:fourier-def}
\end{equation}
The $L^2(\mathbb{R})$ norm $\Vert \cdot \Vert$ is denoted by
\begin{equation}
    {\Vert f \Vert}^2 = \intinf \dd t  \, {|f(t)|}^2.
\end{equation}
The $k$th moment of a real-valued function~$g$, with $k=0,1,2,\ldots$, is denoted by  
\begin{equation} \label{eq: moment def}
    M_k[g] = \int_{-\infty}^\infty \dd t \, t^k g(t).
\end{equation}
Finally, we recall that $\sinc$ denotes the analytic function defined by 
\begin{equation}
    \sinc z=\begin{cases} \frac{\sin z}{z} & z\neq 0\\ 1 & z=0.\end{cases}
\end{equation}

\section{Preliminaries} \label{sec: Preliminaries}

In this section we set out preliminaries about the field-detector model in $d$ spacetime dimensions. 
We introduce the effective temperature for a detector in stationary motion, defined via the detailed balance formula, and we review how the long interaction time limit is implemented in the detector's response and in the effective temperature. 

\subsection{Field-detector model} \label{sec: field detector model}

We consider an Unruh-DeWitt (UDW) detector~\cite{Unruh1976, DeWitt1979}, a pointlike two-level quantum system, on a worldline $\mathsf{x}(\tau)$ parametrised by proper time $\tau$ in $d$-dimensional Minkowski spacetime with $d\ge3$. The detector's Hilbert space $\mathcal{H}_D \cong \mathbb{C}^2$ is spanned by the orthonormal basis $\{\ket{0},\ket{E}\}$ whose elements satisfy $H_D\ket{0}=0$ and $H_D\ket{E} = E\ket{E}$, where $H_D$ and $E$ are the detector's Hamiltonian and energy gap, respectively. If $E>0$, $\ket{0}$ and $\ket{E}$ are the ground and excited states, respectively, whereas if $E<0$, these roles are reversed. We assume throughout that $E\neq 0$, since from Section \ref{sec:circmotion} onward we specialise to $2+1$ dimensions, where the detector's transition rate is discontinuous at $E=0$.

We let the detector interact with a real massless scalar field $\phi$ whose free Hilbert space $\mathcal{H}_\phi$ is the Fock space induced by the Minkowski vacuum. The total Hilbert space is $\mathcal{H}=\mathcal{H}_D \otimes \mathcal{H}_\phi$. 
Working in the interaction picture, we take the interaction Hamiltonian to be
    \begin{equation}
        H_{\text{int}}(\tau) = c \chi(\tau)  \mu(\tau) \otimes \phi(\mathsf{x}(\tau)),
    \end{equation}
where $c$ is a real-valued coupling constant, $\chi(\tau)$ is a real-valued switching function that specifies how the interaction is turned on and off, $\mu(\tau)$ is the detector's monopole moment operator, and $\phi(\mathsf{x}(\tau))$ is the field evaluated on the detector's worldline. The monopole moment operator is given by
\begin{equation}
    \mu(\tau) = \ee^{-\ii E\tau}\sigma^- + \ee^{\ii E\tau}\sigma^+,
\end{equation}
where $\sigma^-\ket{E} = \ket{0}$ and $\sigma^+\ket{0} = \ket{E}$. For $E>0$, $\sigma^+$ is the raising operator and $\sigma^-$ is the lowering operator; for $E<0$, these roles are reversed. 

In the present section, Section~\ref{sec: Preliminaries}, we assume  for concreteness that $\chi$ is a smooth function of compact support, 
$\chi \in C^\infty_0(\mathbb{R})$, and not identically vanishing. 
It follows that~$\widehat\chi$, 
where the hat denotes the Fourier transform in the convention~\eqref{eq:fourier-def}, 
is smooth (and in fact real analytic) and falls off faster than any inverse power. 
We shall relax these assumption from Section \ref{sec:assf} onwards. Earlier literature on detectors with compactly supported switching functions includes \cite{Svaiter1992,Higuchi1993,Sriramkumar1994,Fewster:2016ewy}. 

Before the interaction is turned on, the detector is 
prepared in the state $\ket{0}$ and the field is prepared in a state $\ket{\Psi}$ whose Wightman function $\expval{\phi(\mathsf{x})\phi(\mathsf{x'})}{\Psi}$ is a distribution of Hadamard type in the coincidence limit $\mathsf{x'}\to\mathsf{x}$ \cite{Decanini:2005eg}. 
Working within first-order perturbation theory in the coupling constant~$c$, the probability of finding the detector in the state~$\ket{E}$ after the interaction has ceased, 
without observing the final state of the field, 
is proportional to the response function~$\F_\chi(E)$,
\begin{equation} \label{eq: RF non-stationary}
    \F_\chi(E) = \intinf \dd\tau \intinf \dd\tau' \, \chi(\tau)\chi(\tau')\, \ee^{-\ii E(\tau-\tau')} \mathcal{W}(\tau,\tau'),
\end{equation}
where the distribution $\mathcal{W}(\tau,\tau'):= \expval{\phi(\mathsf{x}(\tau))\phi(\mathsf{x}(\tau'))}{\Psi}$ is the 
pull-back of the Wightman function in the state $\ket{\Psi}$ 
to the detector's worldline 
\cite{Unruh1976,DeWitt1979,birrell,Wald:1995yp,Junker:2001gx}. 
In the rest of the paper we suppress the constant of proportionality and we use the terms ``response function'' and ``probability'' interchangeably. 

When $\ket{\Psi}$ is the $d$-dimensional Minkowski vacuum, 
$\W(\tau,\tau')$ is given by
\begin{equation}
    \W(\tau, \tau') = \frac{\Gamma(d/2-1)}{4\pi^{d/2}\left[ (\mathbf{x}-\mathbf{x'})^2-(t-t'-\ii\epsilon)^2\right]^{(d-2)/2}},
\end{equation}
where $\mathbf{x}= \mathbf{x}(\tau)$, $\mathbf{x'} = \mathbf{x}(\tau')$, $t=t(\tau)$ and $t'=t(\tau')$, and the distributional limit $\epsilon \to 0^+$ is understood. For odd~$d$, the denominator has the phases $\ii^{d-2}$ or $(-\ii)^{d-2}$ when $t-t'>0$ or $t-t'<0$, respectively, because $\mathsf{x}(\tau)$ is a timelike curve.

We note that this model, though standard, has some defects if one tries to understand it beyond the context of first order perturbation theory in $3+1$ dimensions (see, e.g., the brief discussion in section 5 of~\cite{EncycMP_Measurement_in_QFT_FewsterVerch2025}) arising from the singular nature of the coupling. The problem can be ameliorated by extending the smearing around the curve, as is done in the rigorous treatment~\cite{DeBievre:2006pys}, 
or by imposing cut-offs~\cite{HuLin:2006}. 
A~recent treatment that renormalises the interaction can be found in~\cite{Sanchez:2025eix}. 
We will restrict to first order perturbation theory in this paper.

\subsection{Stationary response and the long time limit} \label{sec: stat response}

We now specialise to a detector whose response is stationary, in the sense that $\mathcal{W}(\tau,\tau')$, the pull-back of the Wightman function to the detector's trajectory, depends on $\tau$ and $\tau'$ only through their difference $\tau-\tau'$. We may then write 
$\mathcal{W}(\tau,\tau') = \mathcal{W}(\tau-\tau',0)=:\mathcal{W}(\tau-\tau')$, so that all the dependence of the response on the trajectory and on the field's initial state $\ket{\Psi}$ is encoded in the single-variable distribution~$\mathcal{W}$. 
The response function \eqref{eq: RF non-stationary} can then be recast as \cite{Fewster:2016ewy,Bunney:2023ipk}
\begin{equation} \label{eq: stationary RF}
    \F_\chi(E) = \frac{1}{2\pi} \intinf \dd\omega\, \FTChiSq \, \widehat{\W}(\omega + E), 
\end{equation}
where we recall that the hat denotes the Fourier transform in the convention~\eqref{eq:fourier-def}. 
This situation is realised when the trajectory is stationary in the sense of Minkowski geometry, that is, when it is the orbit of a Killing vector that is timelike in at least some neighbourhood of the trajectory~\cite{LetawStationary,Russo:2009yd,Bunney:2023mkh}, 
and the state $\ket{\Psi}$ is invariant under the Poincar\'e transformations generated by the Killing vector. In particular, when $\ket{\Psi}$ is the Minkowski vacuum, the response is stationary in any geometrically stationary motion. 

While $\widehat{\W}$ is in general a distribution, 
with falloff properties that depend on the spacetime dimension, 
we assume from now on that it is a bounded function, and continuous except possibly at zero argument. This will be the case for circular motion in $2+1$ dimensions, which we consider from Section \ref{sec:circmotion} onward. 

Consider now the limit in which the interaction operates for a long time with approximately constant coupling strength. As $\chi$ has by assumption compact support, we may formalise the long time limit by considering a family of switching functions~$\chi_\lambda$, where $\lambda$ is a positive parameter, such that the support of $\chi_\lambda$ is an interval with length proportional to $\lambda$ as $\lambda\to\infty$, 
and $\chi_\lambda$ is approximately constant within this interval as $\lambda\to\infty$, at least in some averaged sense. 
Specific implementations of this limit were introduced in~\cite{Fewster:2016ewy}, 
under the assumption that $\chi_\lambda$ is non-negative, and we will employ these implementations in  Sections \ref{sec:circmotion}
and~\ref{sec: Small gap long time}. 
These implementations satisfy 
\begin{align}
\lambda^{-1} 
{\Vert \chi_\lambda\Vert}^2
\xrightarrow[\lambda\to\infty]{} 
\frac{c_1}{2\pi}\,,
\ \ c_1>0\ \text{constant}, 
\label{eq:key-lambda-L2norm-limit}
\end{align}
and 
\begin{align}
\lambda^{-1}\F_{\chi_\lambda}(E) \xrightarrow[\lambda\to\infty]{} 
c_1 \widehat{\W}(E)\,,
\ \ c_1>0\ \text{constant}, 
\label{eq:key-Flambda-limit}
\end{align}
for each fixed $E\ne0$, using the continuity of $\widehat{\W}$ at nonzero argument. 
This means that the $\lambda\to\infty$ limit 
implements 
the limit of long interaction time, with the parameter $\lambda$ being proportional to the duration of the interaction,
and with $\widehat{\W}(E)$ being proportional to the detector's transition probability per unit time, known as the transition rate. 
From the experimental point of view, each value of $\lambda$ in this implementation should be understood as the duration of an individual experimental run, and the limit $\lambda\to\infty$ can be approached by longer and longer individual runs. 

We shall show in Section 
\ref{sec:assf}, in Proposition~\ref{prop:chilambda}, 
that \eqref{eq:key-Flambda-limit} holds way beyond the specific implementations of~\cite{Fewster:2016ewy}: 
it holds under rather general conditions on the $\lambda\to\infty$ 
asymptotics of $\chi_\lambda$ and its Fourier transform. In particular, $\chi_\lambda$ does not need to have a uniform sign. This generality will provide a key tool for addressing the case of circular motion in $2+1$ dimensions in the later sections.

\subsection{Effective temperature} \label{sec: effective temp}

We now describe a notion of an effective temperature that characterises the detector's response in long time limit implementations in which \eqref{eq:key-Flambda-limit} holds.

For background, recall that when the detector-field interaction is modelled by a Markovian environment, the detector's state in the weak-coupling-with-long-duration limit approaches the density matrix \cite{DeBievre:2006pys,Juarez-Aubry:2019gjw}
\begin{equation} \label{eq: asymptotic state}
    \rho(E) = \frac{1}{1+\ee^{-E/T(E)}}\begin{pmatrix} 1 & 0 \\
    0 & \ee^{-E/T(E)}\end{pmatrix},
\end{equation}
where $T(E)$ satisfies 
\begin{equation} \label{eq: W hat detailed balance}
    \widehat{\W}(E) = \ee^{-E/T(E)}\widehat{\W}(-E), 
\end{equation}
and is explicitly given by 
\begin{equation} \label{eq: infinite time T}
    \frac{1}{T(E)} = \frac{1}{E}\log\!\left( \frac{\widehat{\W}(-E)}{\widehat{\W}(E)} \right).
\end{equation}
For uniform linear acceleration with proper acceleration~$a$, with the field initially prepared in the Minkowski vacuum, $T(E)$ is independent of the energy $E$ and equal to the Unruh temperature $T_U=a/(2\pi)$ \cite{Unruh1976,DeWitt1979}. 
The state $\rho(E)$ \eqref{eq: asymptotic state} in this case is thus a genuine thermal Gibbs state in the temperature~$T_U$, and 
\eqref{eq: W hat detailed balance} is the detailed balance relation between excitation and de-excitation rates in thermal equilibrium~\cite{Einstein,terhaar-book}. 
For other types of non-inertial stationary motion, $\rho(E)$ \eqref{eq: asymptotic state} is not a thermal state in the Gibbs sense since $T(E)$ depends on~$E$; however, for sufficiently narrow intervals in~$E$, we may view $\rho(E)$ as an approximate Gibbs state in the energy-dependent temperature $T(E)$. We refer to $T(E)$ as the (infinite time) detailed balance temperature. Studies of $T(E)$ in a range of situations are given in 
\cite{Good_2020,Biermann:2020bjh,Juarez-Aubry:2019gjw,Bunney:2023ipk,Bunney:2023vyj,Parry:2024jrm,Bunney:2024qrh}.

Now, in our present case of a finite duration interaction, we consider a family $\chi_\lambda$ of switching functions 
for which \eqref{eq:key-Flambda-limit} holds. 
Following~\cite{Fewster:2016ewy}, 
we define the $\lambda$-dependent detailed balance temperature $T_\lambda(E)$ by 
\begin{equation} \label{eq: finite time T}
    \frac{1}{T_\lambda(E)} = \frac{1}{E}\log\!\left( \frac{\F_\lambda(-E)}{\F_\lambda(E)} \right), 
\end{equation}
where 
\begin{align}
\F_\lambda(E) = \F_{\chi_\lambda}(E)/\lambda . 
\label{eq:Flambda}
\end{align}
For fixed~$\lambda$, $T_\lambda(E)$ does not have an interpretation in terms of the Markovian long time limit considered in~\eqref{eq: asymptotic state}, but in the long interaction duration limit in the sense of $\lambda\to\infty$, 
$\F_\lambda(E)$ becomes proportional to $\widehat{\W}(E)$ and $T_\lambda(E)$ approaches $T(E)$ \eqref{eq: infinite time T}. 
We therefore refer to $T_\lambda(E)$ as the (finite time) detailed balance temperature, and we use it to characterise the response of the detector with the finite time switching function $\chi_\lambda$ and the limit of this response when $\lambda\to\infty$.

\section{Asymptotically scaled switching families}
\label{sec:assf}

In this section we provide a new implementation of the long interaction time limit, such that the long time limit \eqref{eq:key-Flambda-limit} holds, controlling in a mathematically precise way the asymptotics of both the switching function and its Fourier transform. This contains as special cases the adiabatic and plateau scalings of \cite{Fewster:2016ewy} but is significantly more general, and in particular allows the switching functions not to have uniform sign.

Recall that in Section \ref{sec: Preliminaries} we assumed $\chi_\lambda$ to have compact support. Recall also that one of the long time implementations in~\cite{Fewster:2016ewy}, called adiabatic scaling, was to set $\chi_\lambda(\tau) = \chi(\tau/\lambda)$, where $\chi$ is a fixed non-negative switching function. Our generalisation is to require $\chi_\lambda$ to asymptote to $\xi(\cdot/\lambda)$ in a suitable sense as $\lambda\to\infty$, where $\xi$ is a function with suitable properties. 

Our construction of $\xi$ relies on the following Proposition, which will be proved 
at the end of this section.

\begin{proposition}\label{prop:chilambda}
    Let $\chi_\lambda\in C^1(\mathbb{R})$ $(\lambda>0)$ be a family of absolutely integrable switching functions
    with the property that $\lambda^{-1}\widehat{\chi}_\lambda(u/\lambda)$ converges almost everywhere in $u$ as $\lambda\to\infty$. Suppose that
    there exists $\eta\in L^2(\mathbb{R},\dd u/(2\pi))$ with the property that, for all sufficiently large $\lambda>0$, one has $|\widehat{\chi}_\lambda(\omega)|\le \lambda \eta(\lambda\omega)$ for almost all $\omega\in\mathbb{R}$. 
    Then there is $\xi\in L^2(\mathbb{R},\dd t)$ so that 
    \begin{equation}\label{eq:xi}
        \widehat{\xi}(u) = \lim_{\lambda\to\infty} \frac{\widehat{\chi}_\lambda(u/\lambda)}{\lambda}
    \end{equation}
    for almost all $u\in\mathbb{R}$.
    Furthermore, if $V:\mathbb{R}^2\to \mathbb{C}$ is any bounded measurable function that is continuous at 
    $(E_*,0)\in \mathbb{R}^2$, 
    then one has a double limit
\begin{equation}\label{eq:Vlimit}
    \frac{1}{2\pi} \intinf \dd\omega\, \frac{|\widehat{\chi}_\lambda(\omega)|^2}{\lambda} \, V(E,\omega) \xrightarrow[E\to E_*]{\lambda\to\infty} \|\xi\|^2 V(E_*,0).
\end{equation}
 In particular, the scaled response function $\F_\lambda$ \eqref{eq:Flambda} obeys
    \begin{equation}
    \label{eq:F-to-xi2What}
       \F_\lambda(E) \xrightarrow[E\to E_*]{\lambda\to\infty}  {\Vert \xi \Vert}^2 \, \widehat{\W}(E_*), 
    \end{equation} 
    for every $E_*\in\mathbb{R}$ at which $\widehat{\W}$ is continuous. This implies that one also has
    \begin{equation}
    \label{eq:F-to-xi2What-fixedE}
    \F_\lambda(E_*)
    \xrightarrow[\lambda\to\infty]{}
    {\Vert \xi \Vert}^2 \, \widehat{\W}(E_*)
    \end{equation}
    for all such~$E_*$.
\end{proposition}

Note that the function 
$\xi$ provided by Proposition \ref{prop:chilambda} may be the zero element of $L^2(\mathbb{R},\dd t)$, 
and in this case the right-hand sides of \eqref{eq:Vlimit}, 
\eqref{eq:F-to-xi2What}
and 
\eqref{eq:F-to-xi2What-fixedE}
all vanish. 
The assumptions of Proposition \ref{prop:chilambda} even allow the special case where every $\chi_\lambda$ is identically vanishing. 
The case of interest for us is when $\xi$ is not the zero element of $L^2(\mathbb{R},\dd t)$. 
We hence introduce the following Definition. 

\begin{definition}\label{def:assf}
(Asymptotically Scaled Switching Family.) 
If a switching function family $\chi_\lambda$ satisfies the assumptions of Proposition \ref{prop:chilambda} and the function 
$\xi \in L^2(\mathbb{R},\dd t)$ provided by Proposition \ref{prop:chilambda} is not the zero element of $L^2(\mathbb{R},\dd t)$, we call $\chi_\lambda$ an Asymptotically Scaled Switching Family (ASSF). 
\end{definition}

An example of an ASSF is $\chi_\lambda(\tau)=\chi(\tau/\lambda)$ where $\chi\in C_0^1(\mathbb{R})$ and not identically vanishing: this is as in the adiabatic scaling of~\cite{Fewster:2016ewy}, but without requiring $\chi$ to have a uniform sign, and relaxing $\chi$ to be $C^1$ rather than smooth. 
The conditions of Proposition \ref{prop:chilambda} hold with $\eta=|\widehat{\chi}|$, and the Proposition provides $\xi=\chi$, with ${\Vert \xi \Vert}^2>0$. In this example $\xi$ is continuous. A second example is the plateau scaling of~\cite{Fewster:2016ewy}, 
which we discuss in Section~\ref{sec: adiab / plat scalings}, showing that the conditions of Proposition \ref{prop:chilambda} again hold and ${\Vert \xi \Vert}^2>0$, 
but in this example $\xi$ is not continuous. 

We remark that in the adiabatic scaling example 
the continuity of $\xi$ implies the pointwise limit 
\begin{equation}
\label{eq:chilambdalimit}
\chi_\lambda(\tau)\xrightarrow[\lambda\to\infty]{} \xi(0)
\end{equation}
for each fixed $\tau\in\mathbb{R}$. 
The pointwise limit \eqref{eq:chilambdalimit} is however not the property responsible for the long time limit~\eqref{eq:F-to-xi2What}, as the coefficient in \eqref{eq:F-to-xi2What} is ${\Vert \xi \Vert}^2$, which concerns the whole profile of~$\xi$, not just~$\xi(0)$. The long time limit \eqref{eq:F-to-xi2What} 
holds even when $\xi(0)=0$ so that the pointwise limit in \eqref{eq:chilambdalimit} vanishes. 
Again, in the plateau scaling example,  
we show in Section \ref{sec: adiab / plat scalings} that $\xi$ is not continuous at $\tau=0$, 
and $\xi(0)$ is hence not defined, but the long time limit \eqref{eq:F-to-xi2What} still holds. 
These observations show that the long time limit \eqref{eq:F-to-xi2What} has come from the averaged sense in which $\chi_\lambda$ asymptotes to~$\xi(\cdot/\lambda)$, not from a pointwise sense. 

We also remark that the functions in an ASSF do not need to have compact support, nor do they need to have a uniform sign. The latter property will be important in the applications in Section~\ref{sec: negative switchings}. 

Under the circumstances described in Section~\ref{sec: stat response}, 
Proposition \ref{prop:chilambda} shows that 
the long time limit 
\eqref{eq:key-Flambda-limit} at fixed $E\ne0$ holds for an ASSF, 
with $c_1 = {\Vert \xi \Vert}^2$. 
When $\widehat{W}$ fails to be continuous at $E=0$, as is the case for circular motion in $2+1$ dimensions, Proposition \ref{prop:chilambda}
does however not immediately yield a simultaneous long-time, low-gap limit, but the Proposition will provide key elements in controlling this double limit, as we will see in Sections \ref{sec: Small gap long time}
and~\ref{sec: negative switchings}. 

For reference, we record three facts concerning ASSFs. First, because $\chi_\lambda$ in an ASSF are absolutely integrable by assumption, their Fourier transforms $\widehat{\chi}_\lambda(\omega)$ are continuous 
and $o(1)$ as $|\omega|\to\infty$.
Second, the hypotheses relating to $\eta$ imply that $\chi_\lambda$ is square-integrable for sufficiently large~$\lambda$. Indeed, applying~\eqref{eq:Vlimit} to a constant function $V$ and using Plancherel's theorem, we have $\|\chi_\lambda\|^2/\lambda\to\|\xi\|^2$ as $\lambda\to\infty$. Third, given that the $\chi_\lambda$ are continuously differentiable, their absolute integrability would be guaranteed if they are also assumed to have compact support. 

We end this section with the proof of Proposition~\ref{prop:chilambda}.
\begin{proof}
Define $\widetilde{\chi}_\lambda(t)=\chi_\lambda(\lambda t)$. 
Then $\widehat{\widetilde{\chi}}_\lambda(u)=\lambda^{-1}\widehat{\chi}_\lambda(u/\lambda)$, so the hypotheses imply that 
$|\widehat{\widetilde{\chi}}_\lambda(u)|\le \eta(u)$ almost everywhere in $u$ for each sufficiently large $\lambda$. By dominated convergence, it follows that the a.e.\ pointwise limit of the $\widehat{\widetilde{\chi}}_\lambda$ as $\lambda\to\infty$
defines an element of $L^2(\mathbb{R},\dd u/(2\pi))$ which we write as the Fourier transform $\widehat\xi$ of some $\xi\in L^2(\mathbb{R},\dd t)$.   
Making a change of variables, one has
\begin{equation}
    \frac{1}{2\pi} \intinf \dd\omega\, \frac{|\widehat{\chi}_\lambda(\omega)|^2}{\lambda} \, V(E,\omega) = \frac{1}{2\pi} \int_{-\infty}^\infty \dd u |\widehat{\widetilde{\chi}}_\lambda(u)|^2 V(E,u/\lambda)
\end{equation}
and since $V$ is bounded, we may use dominated convergence again, together with continuity of $V$ at $(E_*,0)$, to show that the double limit exists and is equal to
\begin{equation} 
    \int_{-\infty}^\infty\frac{\dd u}{2\pi} |\widehat{\xi}(u)|^2 V(E_*,0) = 
       {\Vert \xi \Vert}^2 \, V(E_*,0)
\end{equation}
using Plancherel's theorem.
The particular case \eqref{eq:F-to-xi2What} follows immediately with $V(E,\omega)=\widehat{\mathcal{W}}(E+\omega)$, and the last statement 
\eqref{eq:F-to-xi2What-fixedE}
is just the fact that the existence of a double limit implies the existence of the limit in $\lambda$ with $E$ held fixed. 
\end{proof}

\section{2+1 circular motion} \label{sec:circmotion}

We now specialise to the response of a detector in uniform circular motion in $2+1$ dimensions. 
We first write down the scaled response $\F_\lambda(E)$
\eqref{eq:Flambda} when $\chi_\lambda$ is a general ASSF\null.
We then find the 
long-time limit and the small-gap limit when the limits are taken in succession in each order, under an additional positivity and uniform boundedness assumption on~$\chi_\lambda$. Finally, we isolate the mathematical obstruction to a simultaneous long-time and small-gap limit, showing that this simultaneous limit will need further assumptions on $\chi_\lambda$ and on the relation between the duration and the gap.

\subsection{Response}

We consider a UDW detector in uniform circular motion in 
$(2+1)$-dimensional Minkowski space, and a field prepared in the Minkowski vacuum. 

In an adapted Lorentz frame, the circular motion worldline can be written as
\begin{equation} \label{eq: CM worldline}
    \mathsf{x}(\tau) = \bigl( \gamma \tau, R\cos(\tfrac{\gamma v \tau}{R}),R\sin(\tfrac{\gamma v \tau}{R}) \bigr),
\end{equation}
where $v$ is the orbital speed with $0<v<1$, $\gamma = {(1-v^2)}^{-1/2}$ is the Lorentz factor, and $R>0$ is the radius of the circular motion. The proper acceleration of this worldline is 
\begin{align}
\label{eq:proper-accel}
a= \frac{\gamma^2v^2}{R} . 
\end{align}
Due to the stationarity of the circular motion worldline, the response function is given by \eqref{eq: stationary RF}, where the Fourier transform of the vacuum Wightman function pulled back to the circular motion worldline \eqref{eq: CM worldline} is \cite{Biermann:2020bjh,Parry:2024jrm}
\begin{equation} \label{eq: CM W hat}
    \widehat{\W}(E) = \frac{1}{4} - \frac{1}{2\pi \gamma} \intoinf \dd z \, \frac{\sin(kEz)}{z}\frac{1}{\sqrt{1-v^2\sinc^2 \! z}},
\end{equation}
where $k=2R/(\gamma v)$.

Properties of $\widehat{\W}$ \eqref{eq: CM W hat} were analysed in \cite{Biermann:2020bjh,Parry:2024jrm}. A~property significant in what follows is that $\widehat{\W}$ is bounded, $0\le \widehat{\W} \le \frac12$. 
The bound from below follows from the Bochner--Schwartz theorem, Theorem IX.10 in~\cite{ReedSimon:vol2}, because $\W$ is positive type (cf.\ Theorem 2.1 in \cite{Fewster:2000}) and translationally invariant. The bound from below can also be obtained by more elementary considerations by observing that the response function \eqref{eq: stationary RF} is non-negative for all~$\FTChiSq$, by the perturbation theory calculation from which the response function emerges, and noting that $\FTChiSq$ can be arbitrarily sharply peaked. 
To see the bound from above, we note that $0\le \widehat{\W}$ implies that the second term in \eqref{eq: CM W hat} is bounded from below by $-\frac14$, and as the second term is odd in~$E$, the bound from above follows.
$\widehat{\W}(E)$ is continuous at $E\ne0$ but discontinuous at $E=0$, with well-defined nonzero limits as $E\to0$ from above and below, as shown in~\cite{Biermann:2020bjh}; this is characteristic of $2+1$ spacetime dimensions~\cite{Parry:2024jrm}. 

It is convenient to split $\widehat{\W}$ into three terms as
\begin{equation}\label{eq:Wsplitup}
    \widehat{\W}(E) = \frac{1}{4} - \frac{k}{2\pi\gamma} U(E)- \frac{\sgn(E)}{4\gamma},
\end{equation}
where 
\begin{equation}
    U(E) =   E\intoinf \dd z \,  \sinc(kEz)\left(\frac{1}{\sqrt{1-v^2\sinc^2 \! z}}-1\right)
\end{equation}
is a bounded odd function (using boundedness of~$\widehat{\W}$). 
In fact, $U$ is continuous everywhere, including $E=0$, 
by dominated convergence and continuity of the $\sinc$ function. 
In the split~\eqref{eq:Wsplitup}, we have subtracted from $\left(1-v^2\sinc^2 \! z\right)^{-1/2}$ its asymptotic large $z$ behaviour, equal to~$1$; the compensating term can be evaluated by the standard formula
$\int_0^\infty\dd z \, \frac{\sin(a z)}{z} = \tfrac12 \pi \sgn(a)$.

Inserting \eqref{eq:Wsplitup} into \eqref{eq: stationary RF} and~\eqref{eq:Flambda}, we obtain
\begin{equation} 
\label{eq:Fsplitup}
   \F_\lambda(E) 
   = \frac{1}{4\lambda} {\Vert \chi_\lambda \Vert}^2 - \frac{kE}{4\pi^2\gamma} \int_{-\infty}^{\infty} \dd\omega\, \frac{|\widehat{\chi}_\lambda(\omega)|^2}{\lambda} V(E,\omega) 
   -  \frac{\sgn(E)}{8\pi\gamma\lambda} \int_{-|E|}^{|E|}\dd \omega \, |\widehat{\chi}_\lambda(\omega)|^2 
\end{equation}
where, in the second and third terms, we have used the fact that $\FTChiSqlambda$ is even in $\omega$ because $\chi_\lambda$ is real-valued, and defined
\begin{equation}\label{eq:Vdef}
    V(E,\omega) = \frac{U(E+\omega)+U(E-\omega)}{2E} = 
    \intoinf \dd z \, \sinc(kEz)\cos(k\omega z)\left(\frac{1}{\sqrt{1-v^2\sinc^2 \! z}}-1\right).
\end{equation}
As the portion of the integrand in parentheses in \eqref{eq:Vdef} is absolutely integrable and the remaining factors are bounded and jointly continuous in $E$ and~$\omega$, it follows that $V$ is continuous and bounded on~$\mathbb{R}^2$. Furthermore, $V(E,\omega)$ is even in both arguments separately. Thus the first term in \eqref{eq:Fsplitup} is even in~$E$, while the second and third terms are odd.

\subsection{Successive long time and small gap limits}
\label{subsec:successive-limits}

We assume the switching function family $\chi_\lambda$ to be an ASSF in the sense of Definition \ref{def:assf}, 
so that the 
limit property \eqref{eq:key-Flambda-limit} holds as $\lambda\to\infty$. In addition, we assume that $\chi_\lambda$ are non-negative and bounded uniformly in~$\lambda$: $0\le \chi_\lambda \le M$, where $M$ is a positive constant, independent of~$\lambda$. 

\subsubsection{$\lambda\to\infty$ followed by $E\to0$}

When $E\ne0$ is fixed, the relevant notion of temperature in the $\lambda\to\infty$ limit is the infinite time detailed balance temperature 
$T(E)$~\eqref{eq: infinite time T}, based 
on the transition rate $\widehat{\W}(E)$~\eqref{eq: CM W hat}. In the subsequent small gap limit, we have \cite{Biermann:2020bjh} 
\begin{equation} \label{eq: W hat small E}
    \widehat{\W}(E) = \frac{\gamma-\sgn(E)}{4\gamma} + O(E), 
\end{equation}
and hence 
\begin{equation} \label{eq: 2+1 CM temp vanishes}
    T(E) = \frac{|E|}{\displaystyle{\log\!\left(\frac{\gamma+1}{\gamma-1}\right)}} + O(E^2).
\end{equation}

We see from \eqref{eq: 2+1 CM temp vanishes} that the temperature vanishes as $O(E)$ as $E\to 0$. The infinite time detailed balance temperature is thus not a good quantifier of the effect of acceleration on the detector's response in the small-gap limit. 
The mathematical reason for the vanishing small-gap temperature is the discontinuity of $\widehat{\W}(E)$ at $E=0$, seen in~\eqref{eq: W hat small E}. This discontinuity arises from the weak decay of $\W(\tau)$ at $\tau\to\pm\infty$, a phenomenon characteristic of circular motion in $2+1$ spacetime dimensions~\cite{Parry:2024jrm}. 

\subsubsection{$E\to0$ followed by $\lambda\to\infty$}

Suppose that $\lambda$ is fixed, and 
consider 
$\F_\lambda(E)$ \eqref{eq:Fsplitup} 
as $E\to0$. 
The first term in \eqref{eq:Fsplitup} is independent of~$E$.
In the third term, we observe that 
$|E|^{-1}\int_{-|E|}^{|E|}\dd\omega |\widehat{\chi}_\lambda(\omega)|^2\to 2|\widehat{\chi}_\lambda(0)|^2$ as $E\to 0$, 
by continuity of $\widehat{\chi}_\lambda$. 
In the second term, the $E\to 0$ limit of the integral may be obtained by dominated convergence, assuming $\lambda$ is so large that $\widehat{\chi}_\lambda$ is square-integrable by Proposition~\ref{prop:chilambda}. 
The result is the small $E$ expansion
\begin{equation}
    \F_\lambda(E) = \alpha_\lambda - \beta_\lambda E + o(E),
\end{equation}
where 
\begin{subequations}
\label{eq:alphabeta}
    \begin{align}
        \alpha_\lambda &=  \frac{1}{4\lambda}{\Vert\chi_\lambda\Vert}^2, \label{eq: alpha}\\
        \beta_\lambda &= \frac{1}{4\pi\gamma\lambda}|\widehat{\chi}_\lambda(0)|^2 + \frac{R}{2\pi^2\gamma^2v\lambda} \int_{-\infty}^{\infty} \dd\omega\, |\widehat{\chi}_\lambda(\omega)|^2 V(0,\omega) ,  \label{eq: beta}
    \end{align} 
\end{subequations}
using
$k=2R/(\gamma v)$.

Note that $\alpha_\lambda$ \eqref{eq: alpha} is positive. In $\beta_\lambda$~\eqref{eq: beta}, the first term is positive by the non-negativity assumption on~$\chi_\lambda$. 
The second term in \eqref{eq: beta} 
is positive for sufficiently large~$\lambda$, because it converges to a positive multiple of $\|\xi\|^2 V(0,0)$ as $\lambda\to\infty$ by Proposition~\ref{prop:chilambda}, and 
\begin{equation}
    V(0,0)= \intoinf \dd z \, \left(\frac{1}{\sqrt{1-v^2\sinc^2 \! z}}-1\right) , 
    \label{eq:V00}
\end{equation}
which is clearly positive.
Therefore, for fixed sufficiently large~$\lambda$, the small $E$ expansion of the detailed balance temperature is 
\begin{equation} \label{eq: T = alpha/2 beta}
    T_\lambda(E) = \frac{\alpha_\lambda}{2\beta_\lambda} + o(1).
\end{equation}

Consider now the subsequent $\lambda \to \infty$ limit in~\eqref{eq: T = alpha/2 beta}. 
In this limit, $\alpha_\lambda$ tends to a positive constant, by Proposition~\ref{prop:chilambda}, and the second term in $\beta_\lambda$ \eqref{eq: beta} tends to a positive constant, as noted above. In the first term in $\beta_\lambda$~\eqref{eq: beta}, we recall that $\chi_\lambda$ are by assumption non-negative and uniformly bounded in~$\lambda$, 
$0\le \chi_\lambda \le M$, where $M$ is a positive constant. From this, an  elementary estimate gives 
$\widehat{\chi}_\lambda(0) \ge M^{-1} {\Vert\chi_\lambda\Vert}^2$. The limit \eqref{eq:key-lambda-L2norm-limit} then shows that the first term in \eqref{eq: beta} is bounded below by $\lambda$ times a positive constant as $\lambda\to\infty$. 

Collecting, we see that the leading term in $T_\lambda(E)$ \eqref{eq: T = alpha/2 beta} vanishes as $O\bigl(\lambda^{-1}\bigr)$ 
as $\lambda \to \infty$. 
This shows that if we take the long time limit after the small gap limit, the detailed balance temperature \eqref{eq: finite time T} fails to serve as a good quantifier of how acceleration affects the detector’s response in the small gap limit.

\subsubsection{Summary}

In summary, when the infinite time limit and the small gap limit are taken in succession, in either order, 
for ASSFs that satisfy the positivity and boundedness assumptions stated in the first paragraph of Section~\ref{subsec:successive-limits}, 
the resulting temperature vanishes. These successive limits are hence not good quantifiers of how acceleration affects the detector’s response in the small gap limit.

\subsection{Obstruction to a simultaneous small gap and long time limit} \label{subsec:obstruction}

We now consider the simultaneous small gap and long time limit. 

For $\F_\lambda(E)$ \eqref{eq:Fsplitup} with a general ASSF, we 
have already observed that the discontinuity of $\widehat{\W}(E)$ at $E=0$ prevents us from reading off from Proposition \ref{prop:chilambda} a 
double limit for $\F_\lambda(E)$ as $E\to 0$ and $\lambda\to\infty$. The first and second term in \eqref{eq:Fsplitup} do however have a defined behaviour in this limit: the first term converges to $\|\xi\|^2/4$, while the integral in the second term has a nonzero double limit by Proposition~\ref{prop:chilambda}, recalling that $V$ is continuous and bounded with $V(0,0)>0$. Thus we have
\begin{equation}\label{eq:Fsplitup_dlimit}
  \F_\lambda(E) 
   = \frac{\|\xi\|^2}{4}(1+o_e(1))   - \frac{kE}{2\pi\gamma} \|\xi\|^2 V(0,0)(1+o_e(1)) 
   -  \frac{\sgn(E)}{8\pi\gamma\lambda} \int_{-|E|}^{|E|}\dd \omega \, |\widehat{\chi}_\lambda(\omega)|^2 
\end{equation} 
as $E\to 0$ and $\lambda\to\infty$ simultaneously, where the subscript $e$ on the $o_e(1)$ terms indicates that they are even in~$E$. In fact the first $o_e(1)$ term in \eqref{eq:Fsplitup_dlimit} is $E$-independent, and the evenness of the second $o_e(1)$ term follows because the second term in \eqref{eq:Fsplitup} is odd.

It is the third term in \eqref{eq:Fsplitup_dlimit} that does not have a defined behaviour as $E\to 0$ and $\lambda\to\infty$ simultaneously: the term depends on the magnitude of $|\widehat{\chi}_\lambda(\omega)|^2/\lambda$ over the interval $-|E| \le \omega \le |E|$ as $E\to 0$ and $\lambda\to\infty$, and this is sensitive to how $|E|$ and $\lambda$ are related to each other in the double limit, as the examples in Sections \ref{sec: Small gap long time} and \ref{sec: negative switchings} show. 

Controlling the third term in \eqref{eq:Fsplitup_dlimit} will be the key in the switching family constructions of Section \ref{sec: negative switchings} in which a positive temperature will be obtained in the small gap limit.

\section{Interlude: adiabatic and plateau scaled switching functions} \label{sec: adiab / plat scalings}

In this section we recall the definitions of two ASSF subfamilies introduced in~\cite{Fewster:2016ewy}, 
the adiabatic scaled family and the plateau scaled family. 
To avoid cluttering the notation, we denote each family by~$\chi_\lambda$, with the family being understood from the context. 

\subsection{Adiabatic scaling}
\label{subsec:adiabatic-scaling-def}

We define an adiabatic scaled switching function family by 
\begin{equation} \label{eq: adiabatic scaled switching function}
    \chi_\lambda(\tau) =  \chi(\tau/\lambda) , 
\end{equation}
where we recall that $\lambda>0$ is the dimensionless scaling parameter. 
We take the unscaled switching function $\chi$ 
to be in $C^1_0(\mathbb{R})$, 
non-negative, and not identically vanishing. This is as in \cite{Fewster:2016ewy} but relaxing $C^\infty_0(\mathbb{R})$ to~$C^1_0(\mathbb{R})$.

As noted in Section~\ref{sec:assf}, this scaling is an ASSF, 
and \eqref{eq:F-to-xi2What-fixedE} gives 
\begin{equation} 
\label{eq: scaled RF adiabatic limit}
     \F_\lambda(E) \xrightarrow[\lambda\to\infty]{} {\Vert \chi \Vert}^2 \, \widehat{\W}(E), 
\end{equation}
for each fixed~$E\neq 0$. 
Moreover, $\chi_\lambda$ are bounded uniformly in~$\lambda$. The family $\chi_\lambda$ hence satisfies the conditions in the first paragraph of Section~\ref{subsec:successive-limits}, 
and the results of Section \ref{subsec:successive-limits} about successive long time and small gap limits apply.

\subsection{Plateau scaling}
\label{subsec:plateauscaling}

We define a plateau scaled switching function family by
\begin{equation} \label{eq: plat scaling}
    \chi_\lambda(\tau) = \int_{-\infty}^\tau \dd\tau' \, \big( \psi(\tau')-\psi(\tau'-\tau_s-\lambda\tau_p) \big),
\end{equation}
where $\tau_s$ and $\tau_p$ are positive constants, 
$\psi\in C_0(\mathbb{R})$,
and $\lambda>0$ is again the dimensionless scaling parameter.
$\widehat{\psi}$~is smooth, it has falloff 
$\widehat{\psi}(\omega) = o(1)$ as $|\omega| \to \infty$, and it is square integrable. 
We assume $\psi$ to be non-negative, not identically vanishing, and with support contained in an interval of duration~$\tau_s$. 
This is as in \cite{Fewster:2016ewy} but relaxing 
$\psi$ to be $C_0(\mathbb{R})$ rather than $C^\infty_0(\mathbb{R})$. 

It follows that 
each $\chi_\lambda$ is in~$C^1_0(\mathbb{R})$, and the family $\chi_\lambda$ is bounded uniformly in~$\lambda$. Each 
$\chi_\lambda$~is switched on over an interval of duration~$\tau_s$, it then stays at constant positive value $\int_{-\infty}^{\infty} \dd\tau' \, \psi(\tau')$ for an interval of duration~$\lambda\tau_p$, and it is switched off over an interval of duration~$\tau_s$. The profiles in the switch-on and switch-off intervals are determined by~$\psi$. 
The scaling parameter $\lambda$ stretches only the plateau where $\chi_\lambda$ has a constant positive value, proportionally to~$\lambda$, 
while the switch-on and switch-off intervals remain of fixed duration.

We shall now show that $\chi_\lambda$ is an ASSF\null.

A~calculation from \eqref{eq: plat scaling} shows that  
\begin{equation}\label{eq:plateu-chilamblahat}
    \widehat{\chi}_\lambda(\omega) = \frac{1-\ee^{-\ii\omega(\tau_s+\lambda\tau_p)}}{\ii\omega}\widehat{\psi}(\omega)\,,
\end{equation}
where the formula is understood in the limiting sense at $\omega=0$. From \eqref{eq:plateu-chilamblahat} we find
\begin{equation}
\left|{\lambda}^{-1}\widehat{\chi}_\lambda(u/\lambda)
\right|^2 = {(\tau_p+\tau_s/\lambda)}^2
\sinc^2 \bigl(\tfrac{1}{2}(\tau_p+\tau_s/\lambda)u\bigr) |\widehat{\psi}(u/\lambda)|^2.
\label{eq:plateau-prop-chihatest}
\end{equation}
As $\psi$ is integrable, $|\widehat{\psi}(\omega)|\le b_1=\intinf\dd t\,|\psi(t)|>0$. 
As $\sinc$ is continuous and satisfies $\sinc^2 \! x \le x^{-2}$ for $x\ne0$, there exists a constant $b_2>1$ such that
$\sinc^2 \! x \le b_2^2 / (1 + x^2)$ for $x\in\mathbb{R}$. 
This gives the bound
$\sinc^2 \bigl(\tfrac{1}{2}(\tau_p+\tau_s/\lambda)u\bigr)
\le b_2^2 / \bigl(1 + \frac{1}{4}(\tau_p+\tau_s/\lambda)^2 u^2 \bigr)
\le b_2^2 / \bigl(1 + \frac{1}{4}\tau_p^2  u^2 \bigr)$. 
Finally, for $\lambda\ge1$, the factor ${(\tau_p+\tau_s/\lambda)}^2$ in \eqref{eq:plateau-prop-chihatest} is bounded by ${(\tau_p+\tau_s)}^2$. It follows that 
$ \left|{\lambda}^{-1}\widehat{\chi}_\lambda(u/\lambda) \right| \le \eta(u)$ 
for all $u\in\mathbb{R}$ and $\lambda\ge 1$, 
where 
$\eta(u) = b_1 b_2 (\tau_p+\tau_s) \bigl(1 + \frac{1}{4}\tau_p^2  u^2 \bigr)^{-1/2}$. 
As $\eta\in L^2(\mathbb{R},\dd u/(2\pi))$, 
the conditions of Proposition \ref{prop:chilambda} hold. 

From \eqref{eq:xi} and \eqref{eq:plateu-chilamblahat} we have 
\begin{equation}
\widehat{\xi}(u) = 
\lim_{\lambda\to\infty} 
\frac{\widehat{\chi}_\lambda(u/\lambda)}{\lambda}
=
\frac{1-\ee^{-\ii u \tau_p}}{\ii u}\widehat{\psi}(0) \,,
\end{equation}
pointwise for all $u\in\mathbb{R}$, where the formula is understood in the limiting sense at $u=0$. 
As $\psi$ is nonnegative and not identically vanishing, we have $\widehat{\psi}(0)>0$, so $\widehat{\xi}$ is nonvanishing 
as an element of $L^2(\mathbb{R},\dd u/(2\pi))$. 
This shows that $\chi_\lambda$ in an ASSF\null. 

Taking the inverse Fourier transform, we find that $\xi$ is $\widehat{\psi}(0)$ times the characteristic function of $[0,\tau_p]$. 
Hence 
$\|\xi\|^2 = \tau_p |\widehat{\psi}(0)|^2$, and from \eqref{eq:F-to-xi2What-fixedE} we have 
\begin{equation} 
\label{eq: scaled RF plateau limit}
    \mathcal{F}_\lambda(E) \xrightarrow[\lambda\to\infty]{}
    \tau_p |\widehat{\psi}(0)|^2 \widehat{\W}(E), 
\end{equation}
for each fixed~$E\ne0$.

As the family $\chi_\lambda$ is uniformly bounded in~$\lambda$, this family 
satisfies the conditions in the first  paragraph of Section~\ref{subsec:successive-limits}, 
and the results of Section \ref{subsec:successive-limits} about successive long time and small gap limits apply. 

Note that $\xi$ is discontinuous at $\tau=0$ and at $\tau=\tau_p$ and is undefined at these points. 
No physical pathology associated with this discontinuity is known to us.

\subsection{Comments}

Our assumptions about the adiabatic and plateau scalings here are chosen so that they will suffice for the 
joint small $E$ and long time analysis in Section~\ref{sec: Small gap long time}, and they involve weaker differentiability than the large $|E|$ analysis in~\cite{Fewster:2016ewy}. 
Yet weaker assumptions within the adiabatic and plateau scalings will be considered in Section~\ref{sec: adiabatic vs plateau new scaling}.

\section{Small gap with long time: double limit with adiabatic and plateau scalings} \label{sec: Small gap long time}

In this section we consider the $2+1$ circular motion response in the simultaneous small gap and long interaction time 
limit, for the adiabatic and plateau scalings of Section \ref{sec: adiab / plat scalings}. We further assume that the simultaneous limit is taken so that $E\to 0$ and $\lambda=\lambda(E)\to\infty$ where $\lambda(E)$ is an inverse power law, given by
\begin{equation} \label{eq: lambda = (S/E)^p}
    \lambda(E) = {\left(\frac{S}{|E|} \right)}^{\!\alpha},
\end{equation}
where $\alpha$ is a positive parameter that specifies the exponent in the power law, and $S$ is a positive constant with units of energy, introduced for consistency with our convention that the scaling parameter $\lambda$ is dimensionless. In an experimental setting, such as an analogue spacetime system \cite{Gooding:2020scc,Bunney:2023ude,Gooding:2025tfp}, 
$S$~would be chosen as an energy scale characterising the physical system.

\subsection{Adiabatic scaling} \label{sec: adiabatic scaling}

For the adiabatic scaled switching function family~\eqref{eq: adiabatic scaled switching function}, we have $\|\chi_\lambda\|^2/\lambda =\|\chi\|^2$ for all $\lambda$ and so the first error term in~\eqref{eq:Fsplitup_dlimit} vanishes, giving
\begin{equation} \label{eq: scaled RF adiabatic fn of E} 
  \F_{\lambda(E)}(E) 
   = \frac{1}{4}\|\chi\|^2   - \frac{kE}{2\pi\gamma} \|\chi\|^2 V(0,0)(1+o_e(1)) 
   -\frac{\sgn(E)}{8\pi\gamma} \int_{-\eta(|E|)}^{\eta(|E|)}\dd \omega \, \FTChiSq , 
\end{equation}
where in the last term we have changed variables, used $k=\frac{2R}{\gamma v}$, and written 
$\eta(|E|) : = S {(|E|/S)}^{1-\alpha}$. The second term in \eqref{eq: scaled RF adiabatic fn of E} 
can be written as $\sgn (E) O(|E|)$, where $O(|E|)$ is even in~$E$, and the notation encodes the evenness in the absolute value in the argument.

The small $E$ behaviour of the last term in \eqref{eq: scaled RF adiabatic fn of E} depends on the small $E$ behaviour of~$\eta(|E|)$, which is determined by~$\alpha$. There are three qualitatively different cases: 
$\alpha>1$, $\alpha=1$ and $0<\alpha<1$, which give the small $E$ behaviour $\eta(|E|)\to\infty$, $\eta(|E|) = S = \text{constant}$ and $\eta(|E|)\to 0$, respectively. We give the small $E$ expansion of this term in Appendix~\ref{app: adiabatic scaling RF expansion}.

Collecting, we find that the small $E$ expansion of \eqref{eq: scaled RF adiabatic fn of E} is
\begin{subequations}
\label{eq: adiab RF allalpha}
    \begin{numcases}{ \F_{\lambda(E)}(E) = } 
        \displaystyle{\frac{1}{4}{\Vert \chi \Vert}^2 \! \left[1 - \gamma^{-1}\sgn(E) \! \left(1 + O(|E|) + o\bigl(|E|^{\alpha-1}\bigr)\right) \right] } & \text{for $\alpha>1$}, \label{eq: adiab RF alpha>1}  \\[1ex]
        \displaystyle{\frac{1}{4}{\Vert \chi \Vert}^2 \! \left[1 - \gamma^{-1}\Xi\sgn(E) \bigl(1 + O(|E|) \bigr) \right] } & \text{for $\alpha=1$}, \label{eq: adiab RF alpha=1} \\[1ex]
        \displaystyle{\frac{1}{4}{\Vert \chi \Vert}^2 -\frac{{|\widehat{\chi}(0)|}^2 {S}^\alpha}{4\pi\gamma} \sgn(E) \! \left( |E|^{1-\alpha} + 
O \bigl( {|E|}^{3(1-\alpha)}\bigr) + O(|E|) \right)} & \text{for $\alpha<1$,}\label{eq: adiab RF alpha<1}
\end{numcases}
\end{subequations}
and the corresponding small $E$ temperatures are 
\begin{subequations}
\label{eq: adiab temp allalpha}
    \begin{numcases}{ T_{\lambda(E)}(E) = } 
        \frac{|E|}{\displaystyle{\log\!\left(\frac{\gamma+1}{\gamma-1}\right)}}\! \left(1 + O(|E|) + o\bigl(|E|^{\alpha-1}\bigr)\right) & \text{for $\alpha>1$}, \label{eq: adiab temp alpha>1}  \\[1ex]
        \frac{|E|}{\displaystyle{\log\!\left(\frac{\gamma+\Xi}{\gamma-\Xi}\right)}}\bigl(1 + O(|E|) \bigr) & \text{for $\alpha=1$}, \label{eq: adiab temp alpha=1}  \\[1ex]
        \displaystyle{\frac{\pi \gamma}{2} \frac{{\Vert \chi \Vert}^2}{{|\widehat{\chi}(0)|}^2} {\left(\frac{|E|}{S}\right)}^{\!\alpha} 
        \! \left( 1 + 
O \bigl( {|E|}^{2(1-\alpha)}\bigr) + O\bigl({|E|}^\alpha\bigr) \right)} & \text{for $\alpha<1$}, \label{eq: adiab temp alpha<1}
\end{numcases}
\end{subequations}
where
\begin{equation}
\label{eq:xi-def}
    \Xi = \frac{1}{2\pi{\Vert\chi\Vert}^2}\int_{-S}^{S}\dd \omega \, \FTChiSq.
\end{equation}
Note that $0<\Xi<1$, because in the $S\to\infty$ limit \eqref{eq:xi-def} becomes unity by the Plancherel theorem. The coefficient of $|E|$ in \eqref{eq: adiab temp alpha=1} is therefore well defined. 
The $O$ and $o$ error terms in 
\eqref{eq: adiab RF allalpha} and \eqref{eq: adiab temp allalpha}
are even in~$E$, as indicated by $|E|$ in their arguments.

From \eqref{eq: adiab temp allalpha} we now see that when $\lambda\to\infty$ and $E\to0$ under the inverse power law relation~\eqref{eq: lambda = (S/E)^p}, 
the temperature vanishes for any~$\alpha>0$.
For $\alpha\geq 1$, the temperature vanishes linearly in~$E$, whereas for $0<\alpha<1$, the temperature vanishes as the weaker power law~${|E|}^\alpha$.  

It is instructive to compare the temperatures \eqref{eq: adiab temp allalpha} to the successive long time and small gap limits of Section~\ref{subsec:successive-limits}. 
For $\alpha>1$, the temperature \eqref{eq: adiab temp alpha>1} agrees with the temperature~\eqref{eq: 2+1 CM temp vanishes}, which was obtained by first letting $\lambda \to \infty$ with $E$ fixed and then letting $E\to0$. 
Conversely, for $\alpha<1$, the power-law ${|E|}^\alpha \propto \lambda^{-1}$ in the temperature \eqref{eq: adiab temp alpha<1} agrees with the power-law $\lambda^{-1}$ in the temperature that was obtained by first letting $E\to0$ with $\lambda$ fixed and then letting $\lambda\to\infty$, as discussed below~\eqref{eq: T = alpha/2 beta}. 
In the intermediate case $\alpha=1$, 
where the interaction duration is inversely proportional to the detector gap, 
the temperature \eqref{eq: adiab temp alpha=1} 
has a linear falloff in terms of $E$ as for $\alpha>1$, and a falloff proportional to $\lambda^{-1}$ in terms of $\lambda$ as for $\alpha<1$. 
These comparisons illustrate the fact that while the interaction duration increases to infinity as $E\to 0$ for all~$\alpha$, the increase is more rapid for large $\alpha$ and less rapid for small~$\alpha$.

\subsection{Plateau scaling} \label{sec: plateau}

For the plateau scaled switching function family~\eqref{eq: plat scaling}, 
we recall 
from Section \ref{subsec:plateauscaling} 
that Proposition \ref{prop:chilambda} applies with $\|\xi\|^2 = \tau_p |\widehat{\psi}(0)|^2$. 
For the scaled response function, 
\eqref{eq:Fsplitup_dlimit} hence gives
\begin{align} 
  \F_{\lambda(E)}(E) 
   & = \frac{\tau_p |\widehat{\psi}(0)|^2}{4}\left(1+o_e(1) 
   - \frac{2kE}{\pi\gamma}   V(0,0)(1+o_e(1))\right) 
   \notag\\
   &\hspace{3ex}
   -  \frac{\sgn(E) \left(\tau_p + \tau_s {(|E|/S)}^\alpha \right)}{4\pi\gamma} 
\int_{-\eta(|E|)}^{\eta(|E|)}\dd u \, \frac{1-\cos u}{u^2}\left|\widehat{\psi}\left(\frac{|E|}{\eta(|E|)}u\right)\right|^2 , 
\label{eq: scaled RF plateau fn of E}
\end{align}
where now $\eta(|E|) := |E| \left( \tau_s + \tau_p {(S/|E|)}^\alpha\right)$, we have used~\eqref{eq:plateu-chilamblahat}, and $k=\frac{2R}{\gamma v}$ as before.  

As with the adiabatic scaling, 
there are three qualitatively different cases: 
$\alpha>1$, $\alpha=1$ and $0<\alpha<1$, which give the small $E$ behaviour $\eta(|E|)\to\infty$, $\eta(|E|) \to \tau_p S$ and $\eta(|E|)\to 0$, respectively. We give the small $E$ expansion of the last term in \eqref{eq: scaled RF plateau fn of E}
in Appendix \ref{app: [plateau] scaling RF expansion}.

Collecting, we find that the small $E$ expansion of \eqref{eq: scaled RF plateau fn of E} is 
\begin{subequations}
\label{eq: plat RF allalpha}
    \begin{numcases}{ \F_{\lambda(E)}(E) = } 
         \displaystyle{\frac{\tau_p}{4}|\widehat{\psi}(0)|^2 
         \Bigl[ 
         \bigl(1 + o_e(1) \bigr) 
         - \gamma^{-1}\sgn(E) \bigl(1 + o_e(1) \bigr) 
         \Bigr]} & \text{for $\alpha>1$}, \label{eq: plat RF alpha>1}  \\[1ex]
        \displaystyle{\frac{\tau_p}{4}|\widehat{\psi}(0)|^2 
         \Bigl[ 
         \bigl(1 + o_e(1) \bigr) 
         - \gamma^{-1}\zeta\sgn(E) \bigl(1 + o_e(1) \bigr) 
         \Bigr]} & \text{for $\alpha=1$}, \label{eq: plat RF alpha=1} \\[1ex]
       \displaystyle{\frac{\tau_p}{4}|\widehat{\psi}(0)|^2 
       \! 
         \left[ 
         \bigl(1 + o_e(1) \bigr) 
         - \frac{\tau_p S}{\pi\gamma}  \left(\frac{|E|}{S}\right)^{1-\alpha} \sgn(E)
         \bigl(1 + o_e(1) \bigr) 
         \right]} & \text{for $\alpha<1$,}\label{eq: plat RF alpha<1}
\end{numcases}
\end{subequations}
and the corresponding small $E$ temperatures are
\begin{subequations}
\label{eq: plat temp allalpha}
    \begin{numcases}{ T_{\lambda(E)}(E) = } 
        \frac{|E|}{\displaystyle{\log\!\left(\frac{\gamma+1}{\gamma-1}\right)}}
        \bigl(1 + o_e(1) \bigr) & \text{for $\alpha>1$} \label{eq: plat temp alpha>1}  \\[1ex]
        \frac{|E|}{\displaystyle{\log\!\left(\frac{\gamma+\zeta}{\gamma-\zeta}\right)}} 
        \bigl(1 + o_e(1) \bigr) & \text{for $\alpha=1$} \label{eq: plat temp alpha=1}  \\[1ex]
        \displaystyle{\frac{\pi\gamma}{2\tau_p} {\left(\frac{|E|}{S}\right)}^{\!\alpha} 
        \bigl(1 + o_e(1) \bigr) } 
        & \text{for $\alpha<1$,}\label{eq: plat temp alpha<1}
\end{numcases}
\end{subequations}
where
\begin{equation}
\label{eq:zeta-def}
    \zeta = \frac{1}{\pi}\int_{-\tau_p S}^{\tau_pS} \dd u \, \frac{1-\cos u}{u^2} = \frac{2}{\pi}\left(\Si(\tau_p S) + \frac{\cos(\tau_p S)-1}{\tau_p S}\right) , 
\end{equation}
and $\Si(z)=\int_0^z \dd z\,\sin(z)/z$ is the sine integral~\cite{NIST}. 
Note that $0<\zeta<1$, and that $\zeta\to 1$ as $\tau_p S \to\infty$~\cite{NIST}.
The coefficient of 
$|E|$ in \eqref{eq: plat temp alpha=1} is therefore well defined. 

From \eqref{eq: plat temp allalpha} we now see that the temperature again vanishes when 
$\lambda\to\infty$ and $E\to0$ under the inverse power law relation~\eqref{eq: lambda = (S/E)^p}, for any~$\alpha>0$. Furthermore, the power of $|E|$ in the leading term in \eqref{eq: plat temp allalpha} is identical to the power of $|E|$ in the leading term in the adiabatic scaling temperature~\eqref{eq: adiab temp allalpha}, and for $\alpha>1$ the leading terms in \eqref{eq: adiab temp alpha>1}
and 
\eqref{eq: plat temp alpha>1}
are in fact fully identical. 

\subsection{Summary}

To summarise, in the simultaneous long time and small gap limit, under the power-law relation~\eqref{eq: lambda = (S/E)^p}, the effective temperature vanishes, both for adiabatic-scaled and plateau-scaled switching functions. The behaviour of the effective temperature is qualitatively similar for the two scalings.

\section{Positive small gap temperature from interaction sign changes} \label{sec: negative switchings}

In Section \ref{sec: Small gap long time}, we found that in the simultaneous limit of long interaction and small energy gap, 
under the inverse power law relation \eqref{eq: lambda = (S/E)^p}
between the gap and the interaction duration, 
the detailed balance temperature \eqref{eq: finite time T} vanishes for both the adiabatic and plateau scaled switching families that were defined in Section~\ref{sec: adiab / plat scalings}. 
In this section, we show that a key assumption behind this property was the non-negativity of the switching functions, and we show how a positive temperature
in the long-interaction-and-small-gap limit can be recovered by allowing the switching functions to change sign.

In Section \ref{sec: conditions on chi} we work within a general ASSF of Definition~\ref{def:assf}, identifying the key necessary property for recovering a positive temperature in the simultaneous limit of small gap and long time.  
Sections \ref{sec: adiabatic vs plateau new scaling}--\ref{sec: compact support} present examples, both with compact support and with noncompact support, where a positive temperature in the limit is recovered, under the inverse power law relation \eqref{eq: lambda = (S/E)^p}
between the gap and the interaction duration.

\subsection{Recovery of a positive temperature} \label{sec: conditions on chi} 

Let the family $\chi_\lambda$ be an ASSF as given in Definition~\ref{def:assf}. By assumption, $\chi_\lambda$ is absolutely integrable and in~$C^1$. 
By Proposition~\ref{prop:chilambda}, $\chi_\lambda$ is square-integrable for sufficiently large~$\lambda$, and 
\begin{equation}
    \|\chi_\lambda\|^2/\lambda\to \|\xi\|^2 > 0 
\end{equation}
as $\lambda\to\infty$.

We assume that $\lambda$ and $E$ are related by $\lambda = \lambda(E)$, 
where the positive-valued function $\lambda(E)$ is even and satisfies $\lambda(E) \to \infty$ as $E\to0$.  

We wish to recover a positive temperature in the $E\to0$ limit. 
From \eqref{eq: finite time T} 
we see that this happens if and only if 
\begin{align}
\frac{\F_{\lambda(E)}(-E)}{\F_{\lambda(E)}(E)} = 1 + \frac{E}{T_0} + o(E)
\label{eq:FlambdaE(E)-ratio}
\end{align}
as $E\to 0$, where $T_0$ is a positive constant, and we then have $T_{\lambda(E)}(E) = T_0 + o(1)$ as $E\to 0$. 
This means that the crucial issue is the balance of the even and odd parts of $\F_{\lambda(E)}(E)$. 

From \eqref{eq:Fsplitup_dlimit}, using the evenness of $\lambda(E)$, we have the decomposition 
\begin{subequations}
\label{eq:F-evenandodd}
\begin{align}
\F_{\lambda(E)}(E) 
& =
\F^{\textnormal{even}}_{\lambda(E)}(E)+\F^{\textnormal{odd}}_{\lambda(E)}(E) , 
\\
  \F^{\textnormal{even}}_{\lambda(E)} (E) 
   & = \frac{\|\xi\|^2}{4}(1+o_e(1)) 
    , 
    \label{eq:F-even}
   \\
   \F^{\textnormal{odd}}_{\lambda(E)} (E) 
   & = 
   - E\frac{\pi}{4a} \|\xi\|^2 I(v)(1+o_e(1)) 
   -  \frac{\sgn(E)}{8\pi\gamma\lambda(E)} \int_{-|E|}^{|E|}\dd \omega \, |\widehat{\chi}_{\lambda(E)}(\omega)|^2 , 
   \label{eq:F-odd}
\end{align} 
\end{subequations}
where in the first term in \eqref{eq:F-odd} we have written $k=2R/(\gamma v)$ in terms of the proper acceleration $a$ \eqref{eq:proper-accel} as $k=2\gamma v/a$, 
and we have written $V(0,0)$ \eqref{eq:V00} in terms of the function 
\begin{equation}\label{eq: I(v)}
     I(v)=\frac{4v}{\pi^2}\int_0^\infty\dd z \left( \frac{1}{\sqrt{1-v^2\sinc^2 \! z}}-1 \right). 
\end{equation}
By the positivity of $\F^{\textnormal{even}}_{\lambda(E)} (E)$ \eqref{eq:F-even} as $E\to0$, 
the positive temperature condition 
\eqref{eq:FlambdaE(E)-ratio} holds if and only if 
\begin{equation}
    \frac{\F^{\textnormal{odd}}_{\lambda(E)}(E)}{\F^{\textnormal{even}}_{\lambda(E)}(E)} = -\frac{E}{2T_0} +o(E)
\label{eq:FoddFeven-ratio}
\end{equation}
as $E\to 0$. 
By \eqref{eq:F-even} and~\eqref{eq:F-odd}, \eqref{eq:FoddFeven-ratio} holds if and only if 
\begin{equation}
    \F^{\textnormal{odd}}_{\lambda(E)}(E) = -\frac{ \|\xi\|^2}{8T_0}
    E(1+o_e(1)) 
\label{eq:F-odd-wishlist}
\end{equation}
as $E\to 0$.

Comparing now 
\eqref{eq:F-even} and~\eqref{eq:F-odd-wishlist}, we see that a positive temperature $T_0$ is attained if and only if 
\begin{equation}
    \frac{1}{|E|\lambda(E)} \int_{-|E|}^{|E|}\dd \omega \, |\widehat{\chi}_{\lambda(E)}(\omega)|^2 
\end{equation}
has a finite, possibly zero, limit as $E\to 0$.
A~nonzero limit would modify the limiting temperature~$T_0$, so we focus on the case where the limit vanishes. This happens when 
\begin{equation} \label{eq:SFS}
    \frac{1}{\lambda(E)} \int_{-|E|}^{|E|}\dd\omega \, {|\widehat{\chi}_{\lambda(E)}(\omega)|}^2 = o(|E|) \ \ \text{as $E\to0$}.  \hspace{1cm} \text{(SFS)}
\end{equation}
We refer to \eqref{eq:SFS} as the Small Frequency Suppression (SFS) condition. 
When SFS holds, the 
small $E$ expansions of $\F^{\textnormal{odd}}_{\lambda(E)}(E)$ and $T_{\lambda(E)}(E)$ become 
\begin{subequations}
\label{eq:nonzero-F-and-temp}
\begin{align}
    \F^{\textnormal{odd}}_{\lambda(E)}(E) &= - \frac{\pi I(v)\|\xi\|^2}{4a}E(1 + o_e(1)),
\\
    T_{\lambda(E)}(E) &= \frac{a}{2\pi I(v)} + o(1), 
\label{eq: non-zero temp}
\end{align}
\end{subequations}
where we recall that $I(v)$ is given in \eqref{eq: I(v)}  
and $a$ is the proper acceleration~\eqref{eq:proper-accel}. 

We emphasise that the SFS condition \eqref{eq:SFS} is a condition on the last term in~\eqref{eq:Fsplitup_dlimit}. This term in \eqref{eq:Fsplitup_dlimit} is precisely the one that does not have a defined behaviour as $E\to 0$ and $\lambda\to\infty$ independently, for a general ASSF\null.

We summarise. 
When the simultaneous small gap and long duration limit is taken so that $\lambda = \lambda(E)$, where the even positive-valued function $\lambda(E)$ satisfies $\lambda(E) \to \infty$ as $E\to0$, a positive small gap temperature is achieved for an ASSF satisfying the SFS condition~\eqref{eq:SFS}. 
The small gap temperature is given by the leading term shown in~\eqref{eq: non-zero temp}, and this temperature is insensitive to the detail of the switching beyond the SFS condition~\eqref{eq:SFS}. 
The function $I(v)$ in \eqref{eq:nonzero-F-and-temp} encodes the ratio of the small gap temperature and the usual linear acceleration Unruh temperature $a/(2\pi)$. 
We show in Appendix \ref{app: I(v) asymptotics} that $I(v)$ is increasing in $v$ and it has the asymptotic forms $I(v) = \tfrac{1}{\pi}v^3 + O(v^5)$ as $v\to 0$ and $I(v) = -\tfrac{2\sqrt{3}}{\pi^2}\log(1-v) +O(1)$ as $v\to 1$.

What remains to show is that ASSFs satisfying SFS exist for some $\lambda(E)$.
We turn to this question next.

\subsection{Adiabatic and plateau scalings revisited} \label{sec: adiabatic vs plateau new scaling}

In this section we revisit the adiabatic and plateau scaled switching functions of Section 
\ref{sec: adiab / plat scalings} in the light of the SFS condition~\eqref{eq:SFS}. 
We assume $\lambda(E)$ to be the inverse power law \eqref{eq: lambda = (S/E)^p} as in Section~\ref{sec: Small gap long time}. We ask whether the zero temperature conclusion of Section \ref{sec: Small gap long time} can be circumvented by some reasonable modification of the technical assumptions about these switching families. 

\subsubsection{Adiabatic scaling} 
\label{subsubsec:adiab-revisited}

We consider first the adiabatic scaled family $\chi_\lambda = \chi(\tau/\lambda)$, where now $\chi$ is in $C^1(\mathbb{R})$, 
absolutely integrable and square integrable, and not identically vanishing. 
Proposition \eqref{prop:compact-prop} then applies with $\eta = |\chi|$, providing $\xi = \chi$. This shows that $\chi_\lambda$ is an ASSF\null. Compared with the assumptions of Sections \ref{sec: adiab / plat scalings} and~\ref{sec: Small gap long time}, we have replaced compact support by square integrability and dropped the assumption of uniform sign. The assumptions of Section \ref{sec: conditions on chi} hence hold. 

We wish to examine the SFS condition~\eqref{eq:SFS}. As $\widehat{\chi_\lambda}(\omega) = \lambda \widehat\chi(\lambda\omega)$, the left-hand side of \eqref{eq:SFS} equals 
\begin{equation}\label{eq: adiab condition sfs}
    \frac{1}{\lambda(E)} \int_{-|E|}^{|E|}\dd\omega \, {|\widehat{\chi}_{\lambda(E)}(\omega)|}^2 = \int_{-\eta(|E|)}^{\eta(|E|)}\dd u \,{|\widehat{\chi}(u)|}^2 ,
\end{equation}
where we have changed variables by $u = \lambda \omega$, and 
$\eta(E) = S {(|E|/S)}^{1-\alpha}$. 
The expression \eqref{eq: adiab condition sfs} is equal to \eqref{eq: A(E) adiabatic} in Appendix~\ref{app: adiabatic scaling RF expansion}. Adapting the analysis of Appendix 
\ref{app: adiabatic scaling RF expansion} to our present, weaker assumptions, we find  
\begin{equation} \label{eq: adiab condition C 2}
\int_{-\eta(|E|)}^{\eta(|E|)}\dd u \,{|\widehat{\chi}(u)|}^2 \sim 
\begin{cases}
        2\pi {\Vert \chi \Vert}^2& \text{ for } \alpha>1 \\
        2 \pi \Xi {\Vert \chi \Vert}^2& \text{ for } \alpha =1 \\
        4 \pi S {|\widehat{\chi}(0)|}^2
        {(|E|/S)}^{1-\alpha} & \text{ for } 0<\alpha<1 ,
    \end{cases}
\end{equation}
where $\Xi$ is defined in \eqref{eq:xi-def} and satisfies $0<\Xi<1$, 
and for $0<\alpha<1$ we have assumed $\widehat{\chi}(0)\ne0$. The SFS condition \eqref{eq:SFS} is hence not satisfied when $\widehat{\chi}(0)\ne0$. This is consistent with what we found in Section~\ref{sec: adiabatic scaling}. 

When $\widehat{\chi}(0)=0$ and $0<\alpha<1$, however, the term shown in \eqref{eq: adiab condition C 2} vanishes. In this case SFS is satisfied if 
$\widehat{\chi}(\omega) \to 0$ sufficiently rapidly as $\omega \to 0$. For example, if $\widehat{\chi}(\omega) = O(|\omega|^q)$ as $\omega \to 0$ with some positive constant~$q$, then 
\begin{align}
\label{eq:adiab-powerlaw-condition}
\int_{-\eta(|E|)}^{\eta(|E|)}\dd u \,{|\widehat{\chi}(u)|}^2
= O \! \left( {|E|}^{(2q+1)(1-\alpha)}\right)
\end{align}
as $E \to 0$. This satisfies SFS when $q>\frac{\alpha}{2(1-\alpha)}$. 

We conclude that when $0 < \alpha < 1$, a positive temperature at small gap is recovered from the adiabatic scaling if 
$\widehat{\chi}(\omega) \to 0$ sufficiently rapidly as $\omega \to 0$. By continuity of~$\widehat\chi$, we then have 
\begin{equation} 
\label{eq: int chi = 0}
0 = \widehat{\chi}(0) = \intinf \dd\tau \, \chi(\tau), 
\end{equation} 
and indeed $\intinf \dd\tau \, \chi_\lambda(\tau) = 0$ for all~$\lambda>0$. The switching function hence cannot be everywhere non-negative: it must change sign at least once.  We shall give examples of noncompact and compact support in Section~\ref{subsec:adiab-examples}.

We note in passing that if $|\widehat{\chi}(\omega)|\to0$ as $\omega\to0$ but more slowly than any positive power of~$|\omega|$, SFS \eqref{eq: adiab condition sfs} does not hold for the power-law $\lambda(E)$ \eqref{eq: lambda = (S/E)^p}, 
for any value of~$\alpha$, as we show in Appendix~\ref{app: growth of lambda}\null. 
We further give in Appendix \ref{app: growth of lambda} an example in which $|\widehat{\chi}(\omega)|$ is asymptotic to $1/\sqrt{-\log|\omega|}$ as $\omega\to0$, 
and SFS holds when $\lambda(E)\to\infty$ as $E\to0$ so that $\lambda(E) = o(\log|E|)$. 
Situations of this type are however likely to have a weaker decay of $\chi_\lambda(\tau)$ as $\tau \to \pm\infty$, and they are therefore less likely to be of interest for analogue spacetime experiments.

\subsubsection{Plateau scaling} 

We next consider the plateau scaled family $\chi_\lambda$ \eqref{eq: plat scaling} as defined in Section~\ref{subsec:plateauscaling}, 
except that we now 
allow $\psi$ to have a nonuniform sign. 
We continue to assume that $\psi$ has a non-empty support contained in an interval of duration~$\tau_s$. We further continue to assume that $\widehat\psi(0)\ne0$: this is necessary and sufficient for $\chi_\lambda$ to be nonvanishing in the plateau of duration $\tau_p\lambda$, between the switch-on and switch-off intervals of duration~$\tau_s$. 
The assumptions of Section \ref{sec: conditions on chi} hence hold. 

We wish to examine the SFS condition~\eqref{eq:SFS}. 
Using~\eqref{eq:plateu-chilamblahat}, 
the left-hand side of \eqref{eq:SFS} equals 
\begin{equation} \label{eq: plateau condition C}
\frac{1}{\lambda(E)} \int_{-|E|}^{|E|}\dd\omega \, {|\widehat{\chi}_{\lambda(E)}(\omega)|}^2 
= 2 
\left(\tau_p + \tau_s \left(\frac{|E|}{S}\right)^{\! \alpha}\right)
\int_{-\eta(|E|)}^{\eta(|E|)}\dd u \, \frac{1-\cos u}{u^2}\left|\widehat{\psi}\left(\frac{|E|}{\eta(|E|)}u\right)\right|^2 ,
\end{equation}
where we have changed variables by $u = \lambda \omega$, and now 
$\eta(|E|) := |E| \left( \tau_s + \tau_p {(S/|E|)}^\alpha\right)$. 
The expression \eqref{eq: plateau condition C}
is proportional to \eqref{eq: B(E)} in Appendix~\ref{app: [plateau] scaling RF expansion}. 
Adapting the analysis of Appendix 
\ref{app: [plateau] scaling RF expansion} to our present assumptions, we find 
\begin{equation} \label{eq: plateau condition C 2}
\frac{1}{\lambda(E)} \int_{-|E|}^{|E|}\dd\omega \, {|\widehat{\chi}_{\lambda(E)}(\omega)|}^2
\sim  \begin{cases}
        2\pi \tau_p |\widehat{\psi}(0)|^2 & \text{ for } \alpha>1 \\
        2 \pi \zeta \tau_p |\widehat{\psi}(0)|^2 & \text{ for } \alpha=1 \\
        2 \tau_p (\tau_p S) |\widehat{\psi}(0)|^2 {(|E|/S)}^{1-\alpha}& \text{ for } 0<\alpha<1,
    \end{cases}
\end{equation}
where $\zeta$ is defined in \eqref{eq:zeta-def} and satisfies $0<\zeta<1$. 
The SFS condition \eqref{eq:SFS} is hence not satisfied for any $\alpha$, given that $\widehat{\psi}(0)\ne0$ by assumption. This is consistent with what we found in Section~\ref{sec: plateau}. 

We conclude that a plateau scaled switching function cannot recover a positive temperature at small gap, even under the weakened assumptions of the present section. 

\subsection{Adiabatic scaling examples} 
\label{subsec:adiab-examples} 

In this section we present two adiabatic scaled switching families for which the SFS condition \eqref{eq:SFS} holds, and the small gap temperature is hence positive and given by the leading term in~\eqref{eq: non-zero temp}.

\subsubsection{Adiabatic scaling with power-law or Gaussian decay}\label{subsec:adiab-power-or-gaussian-decay}

As a first example, 
we consider the unscaled switching function and its Fourier transform given by 
\begin{subequations}
\label{eq: chi-chihat example 1}
    \begin{align} 
    \chi(\tau)& =   N_q \frac{\sqrt{S}}{2\pi} \, 2^{(q+1)/2}
    \, \Gamma\left(\frac{q+1}{2}\right)  \exp\!\left( -\frac{S^2\tau^2}{2} \right){}_1F_1\left( -\frac{q}{2};\frac{1}{2};\frac{S^2\tau^2}{2} \right), \label{eq: chi example 1} \\
    \widehat{\chi}(\omega) &= \frac{N_q}{\sqrt{S}}
    \left(\frac{|\omega|}{S}\right)^q \exp\!\left( -\frac{\omega^2}{2S^2}  \right), \label{eq: chihat example 1}
    \end{align}
\end{subequations}
where $q$ is a dimensionless positive constant, 
$S$ is a positive constant with units of energy, ${}_1F_1$ is the confluent hypergeometric function of the first kind~\cite{NIST}, 
and $N_q = \sqrt{2\pi/\Gamma\bigl(q+\frac12\bigr)}$.  
\eqref{eq: chi example 1}~may be obtained from \eqref{eq: chihat example 1} by using 3.952.8 in~\cite{Gradshteyn:1943cpj}. 
The normalisation constant $N_q$ has been chosen so that ${\Vert\chi\Vert}^2 = 1$ and ${\Vert\widehat\chi\Vert}^2 = 2\pi$. Plots of $\chi$ and $\widehat\chi$ for selected values of $q$ are shown in 
Figures 
\ref{fig: chi example 1}
and~\ref{fig: chihat example 1}. 
The scaled switching function and its Fourier transform are given by $\chi_\lambda(\tau) =  \chi(\tau/\lambda)$ and $\widehat{\chi}_\lambda(\omega) = \lambda \widehat{\chi}(\lambda\omega)$. 

It follows by elementary considerations that $\chi$ and $\widehat\chi$ satisfy the 
assumptions stated in the first paragraph of Section~\ref{subsubsec:adiab-revisited}. 
In particular, $\chi$ has noncompact support, but it is integrable and square integrable, as is seen from the large $|\tau|$ falloff 
\begin{equation}\label{eq: chi asymp example 1}
\chi(\tau) \sim 
\begin{cases}
{\displaystyle 
N_q \sqrt{S} \,  
\frac{ \, 2^q \Gamma(\frac{q+1}{2})}{\sqrt{\pi}\,\Gamma(-\frac{q}{2})}
\frac{1}{{|S\tau|}^{q+1}}}
& \text{for $q\ne2n$, $n=1,2,\ldots$} \ , 
\\[4ex]
{\displaystyle N_{2n} {(-1)}^n \frac{\sqrt{S}}{\sqrt{2\pi}}   \, {(S\tau)}^{2n} e^{-S^2\tau^2/2}} 
& \text{for $q=2n$, $n=1,2,\ldots$} \ , 
\end{cases}
\end{equation}
where the power-law decay for $q\ne2n$ follows from 13.2.4 and 13.2.23 in~\cite{NIST}, 
and the Gaussian decay for $q=2n$ follows by observing from 13.6.16 in \cite{NIST} that in this case  
\begin{equation} 
\label{eq: chi example 1 even order}
    \chi(\tau) = 
    N_{2n} \frac{\sqrt{S}}{\sqrt{2\pi}} \frac{{(-1)}^n}{2^n} \, 
    \exp\!\left( -\frac{S^2 \tau^2}{2} \right) H_{2n} \! \left( \frac{S\tau}{\sqrt{2}} \right),
\end{equation}
where $H_{2n}$ are the Hermite polynomials. 

\begin{figure}[p] 
    \centering
    \includegraphics[width=0.9\textwidth]{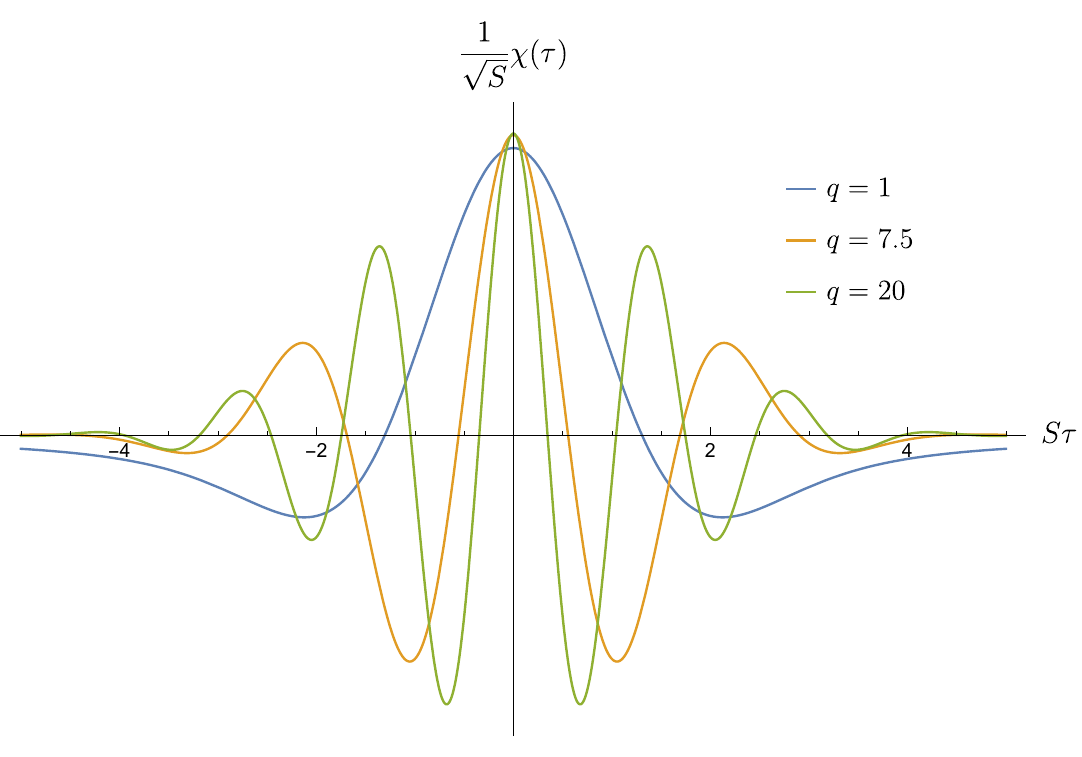}
    \caption{Plots of $\chi(\tau)$ \eqref{eq: chi example 1} for $q=1$, $q=7.5$ and $q=20$. Each curve crosses the horizontal axis at least once at positive $\tau$ and at least once at negative~$\tau$, as is required by the property $\int_{-\infty}^{\infty} \chi(\tau) \, d\tau =0$ and the evenness of~$\chi$.}
    \label{fig: chi example 1}
\end{figure}
\begin{figure}[p]
    \centering
    \includegraphics[width=0.9\textwidth]{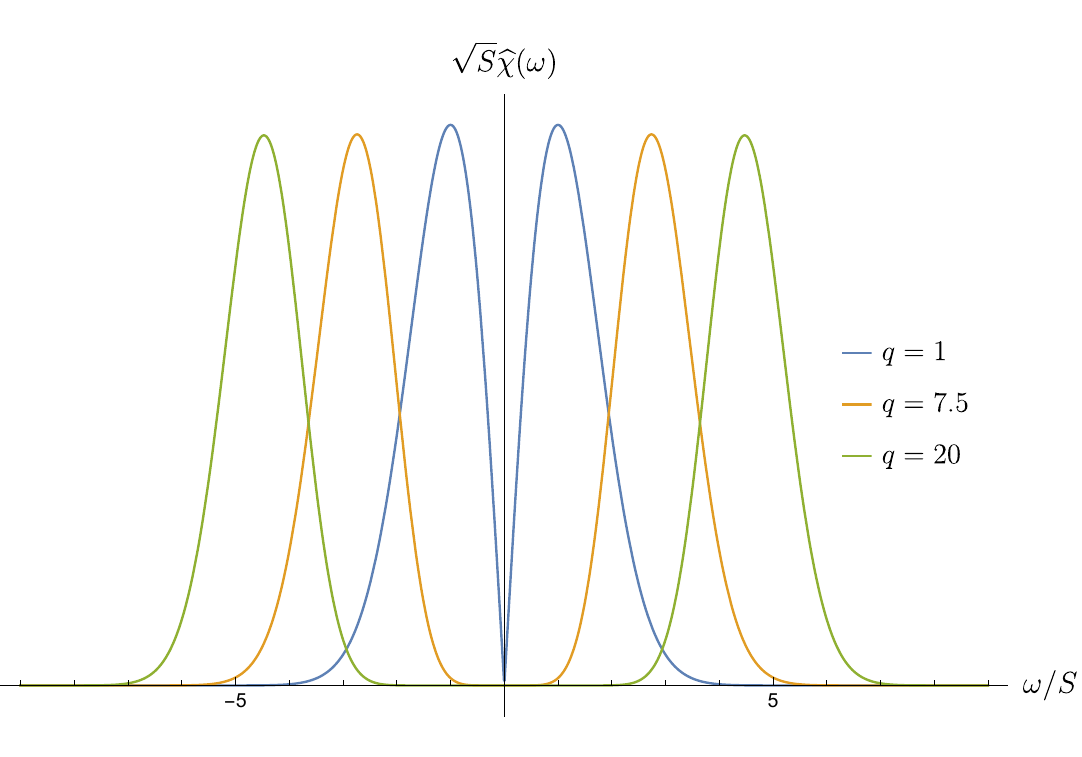}
    \caption{Plots of $\widehat{\chi}(\omega)$ \eqref{eq: chihat example 1} for $q=1$, $q=7.5$ and $q=20$. Each curve displays a characteristic ``double-peak'' structure around the global minimum at $\omega=0$ where $\widehat\chi(0)=0$.}
    \label{fig: chihat example 1}
\end{figure}

Now, to consider SFS~\eqref{eq:SFS}, we note from \eqref{eq: chihat example 1} that  $\widehat{\chi}(\omega) = O(|\omega|^q)$ as $|\omega|\to0$. 
The discussion around \eqref{eq:adiab-powerlaw-condition} then shows that SFS is satisfied provided $q>\frac{\alpha}{2(1-\alpha)}$, where $\alpha$ is the exponent in the power-law scaling 
\eqref{eq: lambda = (S/E)^p} and $0<\alpha<1$. 

The characteristic features of $\chi$ and $\widehat\chi$ are clear from the plots shown in 
Figures 
\ref{fig: chi example 1}
and 
\ref{fig: chihat example 1}. 
$\chi$~has zeroes, as is required by the property $\int_{-\infty}^{\infty} \chi(\tau) \, d\tau =0$, and $\widehat{\chi}$ has a double peak around the origin, as is required by the property $\widehat\chi(0)=0$ and the evenness of~$\widehat\chi$.

\subsubsection{Adiabatic scaling with compact support}\label{subsec:adiab-compact-support}

\begin{figure}[p]
    \centering
    \includegraphics[width=0.9\textwidth]{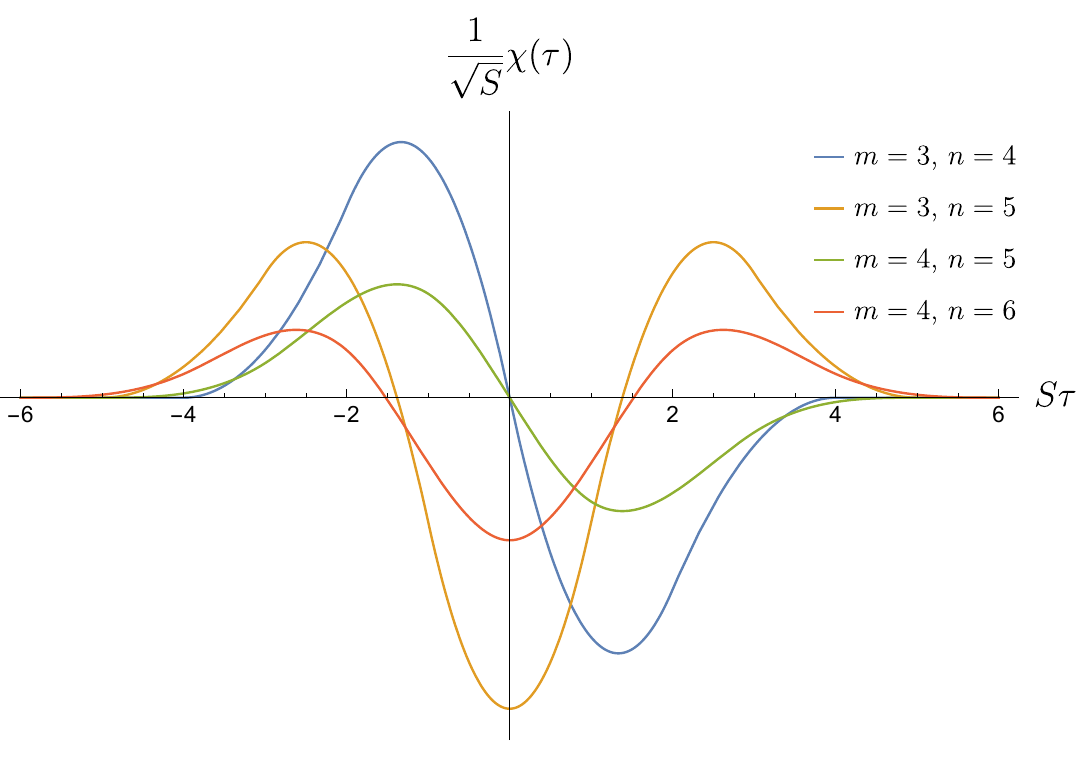}
    \caption{Plots of the compact support switching function $\chi$ \eqref{eq: example 2 adiabatic chi} for selected values of $m$ and $n$. The support $-n \le S\tau \le n$ and and the parity ${(-1)}^{n-m}$ of are clear in the plots.}
    \label{fig: chi example 2} 
\end{figure}

\begin{figure}[p]
    \centering
    \includegraphics[width=0.9\textwidth]{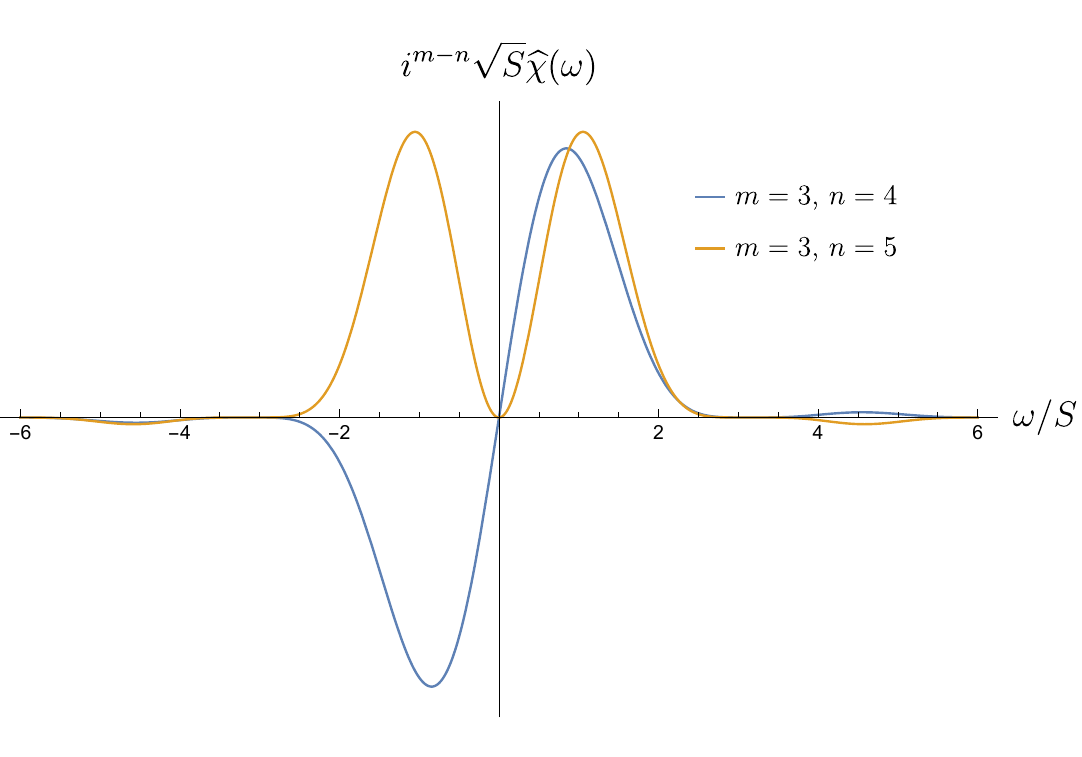}
    \caption{Plots of $\ii^{m-n}\widehat{\chi}$ \eqref{eq: chi hat example 2} that correspond to two of the switching functions plotted in Figure~\ref{fig: chi example 2}. A dominant double peak within $-\pi <  \omega/S < \pi$ is clear in the plots, with power-law suppressed peaks at larger~$|\omega/S|$.}
    \label{fig: chihat example 2}
\end{figure}

As a second example, we consider the unscaled switching function and its Fourier transform given by 
\begin{subequations}
    \begin{align}
        \chi(\tau) &= 
        N_{n,m} \frac{\sqrt{S}}{(m-1)! \, 2^n}\sum_{k=0}^n (-1)^k \binom{n}{k}\Bigl( S\tau+n-2k \Bigr)^{m-1}\Theta\Bigl( S\tau+n-2k \Bigr), \label{eq: example 2 adiabatic chi}\\
         \widehat{\chi}(\omega) &= 
         N_{n,m} \frac{\ii^{n-m}}{\sqrt{S}} \frac{\sin^n\!\left(\frac{\omega}{S} \right)}{\left(\frac{\omega}{S}\right)^m},\label{eq: chi hat example 2}
    \end{align}
\end{subequations}
where $S$ is again a positive constant with units of energy, $n$ and $m$ are integers with $3\le m < n$, and $N_{n,m}$ is a positive constant, for which a cumbersome formula exists, chosen so that ${\Vert\chi\Vert}^2 = 1$ and ${\Vert\widehat\chi\Vert}^2 = 2\pi$. 
Formula \eqref{eq: example 2 adiabatic chi} can be verified from \eqref{eq: chi hat example 2} by contour integration. 

Elementary considerations show that the 
assumptions 
stated in the first paragraph of Section \ref{subsubsec:adiab-revisited} hold; 
in particular, $\chi\in C^1$ follows despite the theta-functions in \eqref{eq: example 2 adiabatic chi} because $m\ge3$. 
For the SFS condition~\eqref{eq:SFS}, 
we see from \eqref{eq: chi hat example 2} that $\widehat{\chi}(\omega) = O({|\omega|}^{n-m})$ as 
$|\omega|\to0$. The discussion around \eqref{eq:adiab-powerlaw-condition} hence shows that the SFS condition is satisfied provided ${n-m>\frac{\alpha}{2(1-\alpha)}}$, where $\alpha$ is the exponent in the power-law scaling 
\eqref{eq: lambda = (S/E)^p} and $0<\alpha<1$. 

The significant feature of this example for us is that $\chi$ has compact support, in the interval $[-n/S, n/S]$. The lower bound of the support is clear from \eqref{eq: example 2 adiabatic chi} because all the theta-functions vanish for $\tau < -n/S$. The upper bound follows from this because $\chi$ has a definite parity, even for $n-m$ even and odd for $n-m$ odd; this is because $\widehat{\chi}$ \eqref{eq: chi hat example 2} is even for $n-m$ even and odd for $n-m$ odd. The upper bound is also seen by observing that for $\tau > n/S$ all theta-functions in \eqref{eq: example 2 adiabatic chi} reduce to unity and using the identity 
\begin{equation} \label{eq: sum vanishes}
    \sum_{k=0}^n(-1)^k \binom{n}{k}\bigl( S\tau+n-2k \bigr)^{m-1}=0,
\end{equation}
which follows from 0.154.6 in~\cite{Gradshteyn:1943cpj}. 

Plots of $\chi$ \eqref{eq: example 2 adiabatic chi} and $\widehat{\chi}$ \eqref{eq: chi hat example 2} for selected values of $n$ and~$m$
are shown in Figures \ref{fig: chi example 2} and~\ref{fig: chihat example 2}.

\subsection{Example: compact support scaling from cropping adiabatic tails} \label{sec: compact support}

The adiabatic scaled switching function examples constructed in Section \ref{subsec:adiab-compact-support} have compact support, and this compact support arose from the special properties of the generalised sinc functions used in~\eqref{eq: chi hat example 2}. In this section we show that compact support switching functions that lead to a positive small gap temperature are much more general: given an adiabatic scaled switching $\chi_\lambda$ that satisfies the assumptions of Section~\ref{sec: conditions on chi}, leading to a positive small gap temperature, and given a slightly stronger large argument falloff condition on~$\chi_\lambda$, 
a compact support modification that leads to a positive small gap temperature under the power-law $\lambda(E)$ \eqref{eq: lambda = (S/E)^p} 
can be found by an appropriate cropping. 

The construction is given in the following Proposition. 

\begin{proposition} 
\label{prop:compact-prop}
Let $\chi$ be a real-valued function in $C^1(\mathbb{R})$, absolutely integrable and square integrable, and not identically vanishing. In addition, let 
\begin{align}
\intinf \dd\tau \, \chi(\tau) = 0 , 
\label{eq:cropped-zeroint}
\end{align}
and 
\begin{align}
\intinf \dd \tau \, |\tau^3 \chi(\tau)| < \infty 
\,.  
\label{eq:thirdmoment-convergence}
\end{align}

Define the switching function families 
$\{\chi_\lambda \mid \lambda>0\}$
and 
$\{\chi_{\lambda,\delta} \mid \lambda>0\}$ by 
\begin{subequations}
\begin{align}
\chi_\lambda (\tau) &= \chi(\tau/\lambda) , 
\label{eq:phi-def-prep} \\
    \chi_{\lambda,\delta}(\tau) &= \chi(\tau/\lambda) f(\tau/\lambda^\delta),
    \label{eq:phi-def}
\end{align}
\end{subequations}
where $f$ is a real-valued smooth function of compact support with $f(0)=1$, and $\delta>1$ is a constant. 

Then: 
\begin{enumerate}
\item 
$\chi_\lambda$ and 
$\chi_{\lambda,\delta}$ are ASSFs. 
\item 
Suppose $\lambda(E)$ is given by the power-law \eqref{eq: lambda = (S/E)^p} with $0<\alpha<2/3$. 
Then $\chi_\lambda$ satisfies the SFS condition~\eqref{eq:SFS}, thus leading to the positive small gap temperature~\eqref{eq: non-zero temp}, 
and the same holds for $\chi_{\lambda,\delta}$ when $\delta > 3/2$. 
\end{enumerate}  
\end{proposition}

\begin{proof}
Given in Appendix~\ref{app: compactsupportaux}. 
\end{proof}

In Proposition \ref{prop:compact-prop}, $\chi_\lambda$ \eqref{eq:phi-def-prep} is an adiabatic scaled family, satisfying the conditions described in Section~\ref{subsubsec:adiab-revisited}. 
For $\chi_\lambda$, the small gap temperature results in the proposition just restate the discussion around \eqref{eq:adiab-powerlaw-condition} for $q=1$. Note that $\chi_\lambda$ \eqref{eq:phi-def-prep} need not have compact support. 

The new part of Proposition \ref{prop:compact-prop} concerns the compact support switching family $\chi_{\lambda,\delta}$~\eqref{eq:phi-def}. 
This family is obtained from $\chi_\lambda$ by cropping off the tails, and because $\delta>1$, the cropping recedes to infinity as $\lambda\to\infty$ faster than the profile of $\chi_\lambda$ spreads: for any fixed~$\tau$, we have $\chi_{\lambda,\delta}(\lambda\tau) = \chi(\tau) f(\lambda^{1-\delta}\tau)\to \chi(\tau)$ as~$\lambda\to0$. 
As $\chi$ changes sign somewhere, this implies that $\chi_{\lambda,\delta}$ changes sign somewhere for all sufficiently large~$\lambda$. 
Proposition \ref{prop:compact-prop} shows that the cropping retains the small gap temperature results when $\delta > 3/2$. 

Stronger assumptions on the large argument falloff of $\chi$ 
and the small argument variation of $f$ give 
variants of Proposition \ref{prop:compact-prop}
with larger ranges of the parameters $\alpha$ and~$\delta$. 
We give an example in Section \ref{app: compactsupportaux-variants} of Appendix~\ref{app: compactsupportaux}.

\section{Conclusions} \label{sec:conclusions}

In this paper, we addressed a puzzle concerning the circular motion Unruh effect for a pointlike Unruh-DeWitt detector interacting with a massless scalar quantum field in $(2+1)$-dimensional Minkowski spacetime. Although the circular motion and linear acceleration Unruh effects are similar over much of the parameter space, the two have a significant discrepancy:
in the limit of long interaction time, 
the effective temperature associated with circular motion, defined operationally via the ratio of the detector's excitation and de-excitation rates, 
vanishes linearly in the gap as the detector's energy gap approaches zero~\cite{Biermann:2020bjh}. 
This property renders the effective temperature an unreliable quantifier of the circular motion Unruh effect at small detector gaps. 
A~practical consequence is that 
proposals to verify the circular motion Unruh effect in analogue spacetime experiments \cite{Retzker:2007vql,Gooding:2020scc,Biermann:2020bjh,Unruh:2022gso,Bunney:2023ude,Gooding:2025tfp} 
would need to operate outside the small gap regime. 

The purpose of this paper was to provide access to the small gap regime in the circular motion Unruh effect in $2+1$ dimensions, by reconsidering how the long interaction time limit is taken. 

As a first step, we introduced a new implementation of the long interaction time limit, the Asymptotically Scaled Switching Family (ASSF) of Definition~\ref{def:assf}, which 
encompasses as special cases both the adiabatic and plateau scalings of \cite{Fewster:2016ewy} but is significantly more general, controlling in a precise way the asymptotics of both the switching function and its Fourier transform. 

We then showed that within the ASSFs, under modest technical uniformity conditions, the Unruh effect remains inaccessible in the small-gap and long-time limit when the two limits are taken in succession in either order, assuming that the coupling of the the local UDW detector to the field does not change sign during the interaction. 

We next turned to the small-gap and long-time limit taken simultaneously, under an inverse power law relation between the two. Within the adiabatic and plateau scalings of~\cite{Fewster:2016ewy}, 
and continuing to assume that the coupling of the the local UDW detector to the field does not change sign during the interaction, we showed that the effective temperature still tends to zero in the small gap limit, although the rate at which the temperature approaches zero can be slowed by taking the interaction duration to increase relatively slowly when the detector gap decreases.

We then observed that recovering a positive small gap temperature is possible only by balancing two conditions on the Fourier transform of the switching function. 
First, the long time limit implies that the squared modulus of the Fourier transform must become increasingly peaked near small values of the argument, but there does not necessarily need to be a peak at the value zero. Second, and pivotally, recovering a finite small gap temperature requires the Fourier transform to vanish sufficiently fast near zero argument, \emph{while simultaneously preserving the sharp peaking of its squared modulus at small arguments implied by the long-time limit\/}. We showed that this cannot be achieved in plateau scaling, but it can be achieved in adiabatic scaling by taking the Fourier transform of the switching function to be strictly vanishing at zero argument. The time average of the switching function then vanishes, so that the switching function must change sign at least once. We gave examples of such adiabatic scaled switching functions with compact support, power-law decay and Gaussian decay. Finally, we presented a construction of sign-changing, compactly supported switching functions that achieve this balance through a scaling that, while not strictly adiabatic, is asymptotically adiabatic at long interaction times. 

Thus, we recover a positive temperature in the small gap and long time limit by arranging the switching function to change sign. However, crucially, the small gap temperature obtained \emph{does not depend on the details of the switching:\/} it depends on the parameters of the motion, in a manner that is consistent with the linear motion Unruh effect prediction, up to a velocity-dependent factor that we identified.
Therefore, the small gap temperature obtained is a property inherent to the system, revealed by our construction of suitably scaled switchings.

In the simultaneous limit of small gap and long interaction time, most of our examples assumed that the gap and the interaction time are related by an inverse power law, and the examples recovering a finite small gap temperature had power-laws where the interaction time goes to infinity more slowly than the inverse gap. 
We discussed in 
Section \ref{subsubsec:adiab-revisited} and Appendix \ref{app: growth of lambda} an adiabatic scaling example where the interaction profile is stretched more slowly than any inverse power of the gap, but we also noted that in such examples the interaction profile itself is likely to have a weak falloff at early and late times, and hence to be of less interest for analogue spacetime experiments. We leave further discussion of non-power-law scalings to future work.

Switching functions that change sign appear not to have featured prominently in the literature on Unruh-DeWitt detectors. 
We note however that such switchings arise naturally in relativistic entanglement-harvesting protocols, where two causally disconnected local laser pulses become entangled with each other through their individual interaction with an electromagnetic field or a Bose-Einstein condensate~\cite{Lindel:2023rfi,Gooding:2023xxl}. 
In this setting, the switching function is given by the coherent amplitude of each laser, which typically changes sign. This underscores the practical relevance and potential experimental feasibility of sign-changing switching functions.

We considered a massless scalar field in $(2+1)$-dimensional Minkowski spacetime because this field appears in the experimental proposals that employ a planar Bose-Einstein condensate or a superfluid Helium surface \cite{Retzker:2007vql,Gooding:2020scc,Biermann:2020bjh,Unruh:2022gso,Bunney:2023ude,Gooding:2025tfp}. 
A massless fermion field in $(2+1)$-dimensional Minkowski spacetime can be simulated in an optical lattice, and this system has been proposed as an analogue spacetime simulation of the linear motion Unruh effect \cite{Rodriguez_Laguna_2017,Kosior_2018}. Were it possible to use the optical lattice for a simulation of the circular motion Unruh effect, such that the signal extraction can be modelled by an UDW detector coupled linearly to the scalar density of the spinor field, the detector analysis would be as for a massless scalar field in $5+1$ dimensions~\cite{Louko:2016ptn}, and there would be no suppression of the signal in the small gap and long time regime~\cite{Parry:2024jrm}. 
A~Maxwell field in $(2+1)$-dimensional Minkowski spacetime is not known to us to arise in an analogue spacetime simulation, but were such a system found, and proposed for simulating the circular motion Unruh effect, such that the signal extraction is modelled by an UDW detector with a dipole interaction, the linear acceleration results in Section 9.3 of \cite{Takagi:1986kn} suggest that the outcomes would be similar to a massless scalar field in $4+1$ dimensions, and hence display no suppression of the signal in the small gap and long time regime~\cite{Parry:2024jrm}.

In conclusion, we have clarified the reasons  behind the vanishing effective Unruh temperature for circular acceleration in the long-time-small-gap limit in $2+1$ spacetime dimensions, highlighting the role that the uniform sign of the switching function plays in this phenomenon. By allowing the switching function to change sign, and carefully controlling the interaction duration, we have shown that it is possible to recover a non-zero effective temperature in the long-time-small-gap limit, with the interaction duration and gap related by an inverse power law. This opens avenues for future theoretical investigations and experimental realisations, suggesting the potential for broader control of detector response and quantum field interactions through tailored switching protocols. Ultimately, our results deepen the conceptual understanding of acceleration-induced quantum effects and offer promising directions for observing analogue spacetime relativistic phenomena in condensed matter experiments.

\section*{Acknowledgements}
We thank Benito Ju\'arez-Aubry, Cameron Bunney, Cisco Gooding, Robert Mann, Rick Perche, Bill Unruh and all the members of the Nottingham Gravity Laboratory for stimulating discussions throughout the course of this work. 
We thank the anonymous referees for helpful comments.
DV-C thanks CONACyT and the WW Smith Trust for financial support. 
The work of CJF was supported by the Engineering and Physical Sciences Research Council [grant number EP/Y000099/1]. 
The work of JL was supported by United Kingdom Research and Innovation Science and Technology Facilities Council [grant numbers ST/S002227/1, ST/T006900/1 and ST/Y004523/1].
The authors have benefited from the activities of COST Action CA23115:
Relativistic Quantum Information, funded  by COST (European Cooperation in Science and Technology). 
For the purpose of open access, the authors have applied a CC BY public copyright licence to any Author Accepted Manuscript version arising.

\appendix

\section{Small $E$ expansion with adiabatic scaling} \label{app: adiabatic scaling RF expansion}

In this appendix we find the small $E$ expansion of the third term in the adiabatic scaled response function~\eqref{eq: scaled RF adiabatic fn of E}. 

We write the third term in \eqref{eq: scaled RF adiabatic fn of E} as 
$-\frac{\sgn(E)}{8\pi\gamma}\mathcal{A}(E)$ where
\begin{equation} 
\label{eq: A(E) adiabatic}
    \mathcal{A}(E) = \int_{-\eta(|E|)}^{\eta(|E|)}\dd \omega \, \FTChiSq , 
\end{equation}
and we recall from Section \ref{sec: adiabatic scaling} that 
\begin{align}
\label{eq:etadisp-adiab}
\eta(|E|)  = S {\left(\frac{|E|}{S}\right)}^{\!1 - \alpha} . 
\end{align}
Note that $\mathcal{A}(E)$ is even in~$E$. 
We consider the cases $\alpha>1$, $\alpha=1$ and $\alpha<1$ in turn.

\subsection{$\alpha>1$} 

Suppose $\alpha>1$. 
Using the evenness of~$\FTChiSq$, we rewrite \eqref{eq: A(E) adiabatic} as
\begin{equation}
\label{eq:A-adiab-int1}
    \mathcal{A}(E) = \intinf\dd \omega \, \FTChiSq - 2\int_{\eta(|E|)}^\infty\dd \omega \, \FTChiSq.
\end{equation}
The first term in \eqref{eq:A-adiab-int1} equals $2\pi {\Vert \chi \Vert}^2$ by Plancherel's theorem. In the second term, we recall that 
$\widehat{\chi}(\omega)=o(\omega^{-1})$ as $\omega \to \infty$, since $\chi \in C_0^1$, 
and we see from \eqref{eq:etadisp-adiab} that $\eta(|E|)$ tends to infinity proportionally to ${|E|}^{-(\alpha-1)}$ as $E\to0$. As $E\to0$, we hence have  
\begin{equation}
    \mathcal{A}(E) = 2\pi {\Vert \chi \Vert}^2 + o\bigl(|E|^{\alpha-1}\bigr),
\end{equation}
where writing the error term as a function of $|E|$ denotes
that the error term is even in~$E$. 

The third term in \eqref{eq: scaled RF adiabatic fn of E} is hence 
\begin{align}
- \frac{{\Vert \chi \Vert}^2}{4\gamma} \sgn(E) 
\! \left(
1 + 
o\bigl(|E|^{\alpha-1}\bigr)
\right) . 
\end{align}

\subsection{$\alpha=1$}

Suppose $\alpha=1$. 
From \eqref{eq:etadisp-adiab} we now have $\eta(|E|) = S$, which is independent of~$E$, and \eqref{eq: A(E) adiabatic} becomes
\begin{equation}
    \mathcal{A}(E) = \int_{-S}^{S}\dd\omega \,  \FTChiSq .
\end{equation}
The third term in \eqref{eq: scaled RF adiabatic fn of E} is hence 
\begin{align}
- \frac{1}{8\pi\gamma} \sgn(E) 
\int_{-S}^{S}\dd\omega \,  \FTChiSq .
\end{align}

\subsection{$\alpha<1$} 

Suppose finally that $\alpha<1$. From \eqref{eq:etadisp-adiab} we now have $\eta(|E|) \to0$ as $E \to 0$. Hence \eqref{eq: A(E) adiabatic} gives 
\begin{equation}
     \mathcal{A}(E) = 2{|\widehat{\chi}(0)|}^2 S^\alpha |E|^{1-\alpha} 
     \! \left(1 + O\bigl(|E|^{2(1-\alpha)}\bigr)
     \right) ,
\end{equation}
where we recall that $\widehat{\chi}(0)>0$ by the assumptions on~$\chi$, and the error term comes using smoothness and evenness of~$\FTChiSq$. 

The third term in \eqref{eq: scaled RF adiabatic fn of E} is hence 
\begin{align}
- \frac{{|\widehat{\chi}(0)|}^2}{4\pi\gamma} 
S^\alpha \sgn(E) |E|^{1-\alpha} 
\! \left(1 + O\bigl(|E|^{2(1-\alpha)}\bigr)
     \right) . 
\end{align}

\section{Small $E$ expansion with plateau scaling}\label{app: [plateau] scaling RF expansion}

In this appendix we find the small $E$ expansion of the third term in the plateau scaled response function \eqref{eq: scaled RF plateau fn of E}.

We write the integral in the third term in \eqref{eq: scaled RF plateau fn of E} as 
\begin{equation} \label{eq: B(E)}
    \mathcal{B}(E):=\int_{-\eta(|E|)}^{\eta(|E|)}\dd u \, \frac{1-\cos u}{u^2}\left|\widehat{\psi}\left(\frac{|E|}{\eta(|E|)}u\right)\right|^2 ,  
\end{equation}
where we recall from Section \ref{sec: plateau} that 
\begin{align}
\label{eq:etadisp-plateau}
\eta(|E|) = |E| \! \left( \tau_s + \tau_p {\left(\frac{S}{|E|}\right)}^{\!\alpha} \,
\right) .  
\end{align}
From \eqref{eq:etadisp-plateau} it follows that 
    \begin{align}
        \frac{|E|}{\eta(|E|)} &= 
        \frac{1}{\tau_s}\frac{c_2{|E|}^\alpha}{1+c_2{|E|}^\alpha} , 
        \label{eq: E/eta}
    \end{align}
where $c_2 = \tau_s/(\tau_pS^\alpha)$. 
Note that $|E|/\eta(|E|) \to 0$ proportionally to ${|E|}^\alpha$ as $E\to0$. 

Note that $\mathcal{B}$ is even. 
We consider the small $E$ behaviour of $\mathcal{B}(E)$ separately for $\alpha>1$, $\alpha=1$ and $0<\alpha<1$.

\subsection{$\alpha>1$}
Suppose $\alpha>1$. Then $\eta(|E|)\to \infty$ as $E\to 0$, and from $\eqref{eq: E/eta}$ we have $|E|/\eta(|E|) \to0$ as $E\to 0$. Using the evenness of $\bigl|\widehat\psi\bigr|^2$, we may write \eqref{eq: B(E)} as 
\begin{equation}
    \mathcal{B}(E) = 2\int_0^{\infty} \dd u \, \frac{1-\cos u}{u^2}\left|\widehat{\psi}\left(\frac{|E|}{\eta(|E|)}u\right)\right|^2 - 2\int_{\eta(|E|)}^\infty\dd u \,\frac{1-\cos u}{u^2}\left|\widehat{\psi}\left(\frac{|E|}{\eta(|E|)}u\right)\right|^2 . 
\label{eq:B(E)-alphalarge}
\end{equation}
As $E\to0$, \eqref{eq:B(E)-alphalarge} gives 
\begin{equation} 
\label{eq: B(E) case I expansion}
    \mathcal{B}(E) = \pi|\widehat{\psi}(0)|^2 + o_e(1) , 
\end{equation}
using in the first term a dominated convergence argument and the standard integral $\int_0^{\infty} \dd u \, \frac{1-\cos u}{u^2} = \tfrac12\pi$, 
and using the second term the boundedness of $\widehat\psi$ and the integrability of 
$(1-\cos u)/u^2$.

Hence, the third term in \eqref{eq: scaled RF plateau fn of E} is
\begin{equation}
- \frac{\tau_p}{4\gamma}|\widehat{\psi}(0)|^2 \sgn(E) \bigl( 1 + o_e(1)\bigr) . 
\end{equation}

\subsection{$\alpha=1$}

Suppose $\alpha=1$. Then $\eta(|E|) \to \tau_p S$ and $|E|/\eta(|E|) \to 0$
as $E\to 0$, and from \eqref{eq: B(E)} we have 
\begin{equation}
\mathcal{B}(E) = \pi |\widehat{\psi}(0)|^2  \zeta + o_e(1), 
\end{equation}
where 
\begin{equation}
    \zeta = \frac{1}{\pi}\int_{-\tau_p S}^{\tau_p S} \dd u \, \frac{1-\cos u}{u^2} . 
\end{equation}
Hence, the third term in \eqref{eq: scaled RF plateau fn of E} is
\begin{equation}
- \frac{\tau_p \zeta}{4\gamma}|\widehat{\psi}(0)|^2 \sgn(E) \bigl(1 + o_e(1) \bigr) . 
\end{equation}

\subsection{$\alpha<1$} \label{sec: plat case III}
Finally, suppose $\alpha<1$. Then 
$\eta(|E|)\to 0$ as $E\to 0$. 

We write \eqref{eq: B(E)} as 
\begin{equation}
    \mathcal{B}(E) = \eta(|E|)\int_{-1}^1 \dd v \, \frac{1-\cos\big(\eta(|E|)v\big)}{\eta^2(|E|)v^2}|\widehat{\psi}(|E|v)|^2, 
\end{equation}
by the change of variables $u= \eta(|E|)v$. 
As $E\to0$, a dominated convergence argument allows the limit to be taken under the integral, with the result 
\begin{align}
\mathcal{B}(E) &= \eta(|E|) |\widehat{\psi}(0)|^2\bigl(1 + o_e(1) \bigr) .  
\end{align}
Hence, the third term in \eqref{eq: scaled RF plateau fn of E} is
\begin{equation}
- \frac{\tau_p (\tau_p S)}{4\pi\gamma} 
|\widehat{\psi}(0)|^2 \left(\frac{|E|}{S}\right)^{1-\alpha} \sgn(E) 
\bigl(1 + o_e(1) \bigr) . 
\end{equation}

\section{Adiabatic scaling with weak small-frequency suppression of  $|\widehat{\chi_\lambda}|$} 
\label{app: growth of lambda}

In this appendix we address the SFS condition \eqref{eq:SFS} for adiabatic scaled switching families, $\chi_\lambda = \chi(\tau/\lambda)$, under the assumptions of Section~\ref{subsubsec:adiab-revisited}, but in the special case where $|\widehat{\chi}(\omega)|\to0$ as $\omega\to0$ more slowly than any positive power of~$|\omega|$. We first show that in this case the SFS condition does not hold for any inverse power-law growth of $\lambda(E)$ as $E\to0$. 
We then give an example in which $|\widehat{\chi}(\omega)|$ is asymptotic to $1/\sqrt{-\log|\omega|}$ as $\omega\to0$, 
and SFS holds when $\lambda(E)\to\infty$ as $E\to0$ so that $\lambda(E) = o(\log|E|)$. 

We formalise the first statement as the following Proposition. 

\begin{proposition} \label{prop: adiab SFS}
Let $\chi_\lambda(\tau) = \chi(\tau/\lambda)$ be an adiabatic scaled switching function family such that $\chi$ is in $C^1(\mathbb{R})$, real-valued, 
absolutely integrable and square integrable, and not identically vanishing. 
Suppose that for every $p>0$, $|u|^{-p}|\widehat{\chi}(u)|^2\to \infty$ as $u \to 0$. 

For $E\ne0$, let $\lambda(E)$ be positive-valued, even in $E$, and such that $\lambda(E)\to\infty$ as $E\to0$. 

Then, if $\chi_{\lambda(E)}(u)$ satisfies the SFS condition~\eqref{eq:SFS}, we have $\lambda(E) = o(|E|^{-\beta})$ as $E\to 0$ for all $\beta>0$.

\end{proposition}

\begin{proof}
Proceeding as in~\eqref{eq: adiab condition sfs}, the SFS condition \eqref{eq:SFS} takes the form  
\begin{equation}
\frac{1}{|E|}\int_{-\eta(E)}^{\eta(E)}\dd u \,{|\widehat{\chi}(u)|}^2 
\xrightarrow[E\to0]{}0,
\label{eq: adiab SFS app}
\end{equation}
where $\eta(E):=\lambda(E)|E|$. Let $p>0$. By the assumptions about $\widehat\chi$, 
there exists a $\delta>0$ and an $M>0$ such that $|\widehat{\chi}(u)|^2 \geq M |u|^p$ when $|u| \leq \delta$. From \eqref{eq: adiab SFS app} it then follows that $\eta(E)\to0$ as $E\to0$. This implies $\lambda(E) = o(|E|^{-\beta})$ as $E\to 0$ for all $\beta\ge1$.

Now, when $|E|$ is so small that $\eta(E) < \delta$, we compute
\begin{align} \label{eq: adiab SFS 2 app}
        \frac{1}{|E|}\int_{-\eta(E)}^{\eta(E)} \dd u \, |\widehat{\chi}(u)|^2 \geq \frac{2M}{|E|}\int_0^{\eta(E)} \dd u \, u^p = \frac{2M}{p+1}\frac{{\eta(E)}^{p+1}}{|E|} =  \frac{2M}{p+1}\lambda(E)^{p+1}|E|^p.
\end{align}
Since \eqref{eq: adiab SFS app} holds, the rightmost expression in \eqref{eq: adiab SFS 2 app} vanishes as $E\to 0$. Hence $\lambda(E)^{p+1}|E|^p\to 0$ as $E\to 0$. 
Since $p>0$ was arbitrary, this implies that  $\lambda(E)|E|^\beta \to 0$ as $E\to 0$ for all $0<\beta<1$. This gives $\lambda(E) = o(|E|^{-\beta})$ as $E\to 0$ for all $0<\beta<1$.

Combining, we have $\lambda(E) = o(|E|^{-\beta})$ for all $\beta>0$.
\end{proof}

As an example, suppose that 
\begin{equation}
    |\widehat{\chi}(u)|^2 = 
    \frac{1}{\log\!\left(S/|u| \right)}
    \bigl(1 + o(1)\bigr)
\label{eq:ex-loginv-as}
\end{equation}
as $u\to 0$, where 
$S$ is a positive constant of the dimension of energy, introduced to make the argument of the logarithm dimensionless. 
For $\lambda(E)$ such that $\eta(E)\to0$ as $E\to0$, we then have 
\begin{align}
\frac{1}{|E|}\int_{-\eta(E)}^{\eta(E)}\dd u \,{|\widehat{\chi}(u)|}^2 
& = 
\frac{2S}{|E|} 
{\textstyle{
E_1 \! \left( \log \bigl(\frac{S}{\eta(E)}\bigr) \right) }}
\bigl(1 + o(1)\bigr)
\notag\\
& = 
\frac{2S}{|E|} 
\! \left( 
\frac{\frac{\eta(E)}{S}}{\log\bigl(\frac{S}{\eta(E)}\bigr)} 
\right)
\bigl(1 + o(1)\bigr)
\notag\\[1ex]
& = 
\frac{2\lambda(E)}{\log(S/|E|) - \log\bigl(\lambda(E)\bigr)}
\bigl(1 + o(1)\bigr) , 
\label{eq:ex-loginv-ascalc}
\end{align}
where in the first equality we have estimated the integral by the exponential integral function $E_1$, and in the second equality we have used the large argument expansion 6.12.1 in~\cite{NIST}. SFS is hence satisfied if $\lambda(E)\to\infty$ as $E\to0$ so slowly that 
$\lambda(E) = o \bigl(\log(S/|E|)\bigr)$. 

This example also illustrates that Proposition 
\ref{prop: adiab SFS} does not have a straightforward converse: under the conditions of the proposition, a slower-than-powerlaw growth of $\lambda(E)$ is necessary for SFS, but it is not sufficient. For example, suppose that $\lambda(E) \sim \bigl(\log(S/|E|)\bigr)^2$ as $E\to0$. 
The growth of $\lambda(E)$ is then slower than any inverse power, but SFS is not satisfied in this example because the last expression in \eqref{eq:ex-loginv-ascalc} tends to infinity as $E\to0$.

\section{Asymptotics of $I(v)$ \eqref{eq: I(v)}} \label{app: I(v) asymptotics}

In this appendix, we find the $v\to 0$ and $v\to 1$ asymptotics of the function $I(v)$ \eqref{eq: I(v)} given by
\begin{equation} \label{eq: I(v) app}
    I(v) = \frac{4v}{\pi^2}\int_0^\infty \dd z \left( \frac{1}{\sqrt{1-v^2\sinc^2 \! z}}-1 \right) . 
\end{equation}

The integrand in \eqref{eq: I(v) app} is positive and pointwise increasing in~$v$. It follows that $I(v)$ is positive and increasing in~$v$. 

Consider first the $v\to 0$ limit. Expanding the integrand in \eqref{eq: I(v) app} in powers of~$v$, justified by dominated convergence, 
we find
\begin{align}
    I(v) = \frac{2v^3}{\pi^2}\int_0^\infty \dd z \, \sinc^2 \! z 
    + O(v^5) 
    = \frac{1}{\pi}v^3
    + O(v^5) . 
\end{align}

Consider then the $v\to 1$ limit. Recalling that $\gamma = 1/\sqrt{1-v^2}$, we write 
\begin{subequations}
    \begin{align}
        I(v) &= \frac{4v}{\pi^2} 
        J \! \left(\frac{1}{\sqrt{1-v^2}}\right), 
        \label{eq:I-ito-J}\\
        J(\gamma) &= \int_0^\infty\dd z \, \left( \frac{1}{\sqrt{1-(1-\gamma^{-2})\sinc^2 \! z}} -1\right). \label{eq: J(gamma)}
    \end{align}
\end{subequations}
$v\to 1$ then corresponds to $\gamma \to \infty$.

We split the integral in \eqref{eq: J(gamma)} at $z=\pi/2$, defining
\begin{subequations}
    \begin{align}
        J(\gamma) &= J_<(\gamma) + J_>(\gamma), \\
        J_<(\gamma) &= \int_0^{\pi/2}\dd z \left( \frac{1}{\sqrt{1-(1-\gamma^{-2})\sinc^2 \! z}} -1\right), \label{eq: J<simp}\\
        J_>(\gamma) &= \int_{\pi/2}^\infty\dd z \left( \frac{1}{\sqrt{1-(1-\gamma^{-2})\sinc^2 \! z}} -1\right) . 
        \label{eq: J>simp}
    \end{align}
\end{subequations}

For $J_>(\gamma)$ \eqref{eq: J>simp}, we have
\begin{equation}
    0\le J_>(\gamma) \le \int_{\pi/2}^{\infty} \dd z \left( \frac{1}{\sqrt{1-(1-\gamma^{-2})/z^2}} -1\right)= \frac{\pi}{2}- 
    \sqrt{\frac{\pi^2}{4} - 1 + \gamma^{-2}} .  
\label{eq:J>bound}
\end{equation}
The rightmost expression in \eqref{eq:J>bound} is bounded as 
$\gamma\to\infty$. Hence $J_>(\gamma) = O(1)$ as $\gamma\to\infty$.

For $J_<(\gamma)$ \eqref{eq: J<simp}, we make a change of variables $u(z)=\sqrt{1-\sinc^2\! z}$, giving
\begin{equation}
    J_<(\gamma) = \int_0^{u(\pi/2)} \dd u \frac{1}{\sqrt{\gamma^{-2}+(1-\gamma^{-2})u^2}}\frac{1}{u'(z(u))}.
\end{equation}
One may show that $\sqrt{3} \le 1/u'(z(u))\le \sqrt{3} +cu^2$ on the integration range, for some constant~$c > 3\sqrt{3}/5$.
Thus, by elementary integration,
\begin{equation}
    \frac{\sqrt{3}}{B}\sinh^{-1}(\gamma AB) \le J_<(\gamma) \le  \frac{\sqrt{3}}{B}\sinh^{-1}(\gamma AB) + \frac{cA^2}{2B},
\end{equation}
where $A=u(\pi/2)$, $B=\sqrt{1-\gamma^{-2}}$, and we have estimated 
$u^2/\sqrt{\gamma^{-2}+ B^2 u^2}\le u/B$.
Noting that 
$B = 1 + O(\gamma^{-2})$ and 
$\sinh^{-1}(\gamma AB)= \log (2\gamma A)+O(\gamma^{-2})$
as $\gamma\to\infty$, this gives 
$J_<(\gamma) = \sqrt{3} \log (\gamma) + O(1)$ as $\gamma\to\infty$. 

Combining these estimates for $J_<$ and $J_>$ gives, as $\gamma \to \infty$, 
\begin{align}
\label{eq:Jgamma-asympt}
    J(\gamma) &= \sqrt{3} \log (\gamma) + O(1)
    \notag\\
    &= 
    -\frac{\sqrt{3}}{2}\log(1-v) + O(1), 
\end{align}
where in the last expression we have written $\gamma$ in terms of $v$ and the limit is $v\to1$.  

Finally, combining \eqref{eq:I-ito-J} and \eqref{eq:Jgamma-asympt} gives 
$I(v) = -\frac{2\sqrt{3}}{\pi^2}\log\bigl(1-v \bigr) + O(1)$ as $v\to1$.

We note in passing that the estimate \eqref{eq:Jgamma-asympt} can be improved to 
\begin{equation}
J(\gamma) = \sqrt{3}\log\bigl(\tfrac{2}{\sqrt{3}}\gamma \bigr) 
+ \int_0^\infty \dd z \left( \frac{1}{\sqrt{1-\sinc^2 \! z}}-\frac{\sqrt{3}}{z(1+z)}-1 \right) 
+ o(1), 
\label{eq:Jgamma-asympt-improved}
\end{equation}
by the technique of Appendix D.2 in~\cite{Biermann:2020bjh}. This allows one to find the constant term in $I(v)$ as $v\to1$.

\section{Proof and variants of Proposition \ref{prop:compact-prop}} \label{app: compactsupportaux}

This appendix provides the proof of Proposition~\ref{prop:compact-prop}, split into Sections \ref{app: compactsupport-chilambda}, \ref{app: compactsupport-part1} and~\ref{app: compactsupport-part2}. 
Section \ref{app: compactsupportaux-variants} discusses variants of Proposition \ref{prop:compact-prop} that ensue under stronger assumptions on the switching functions.

\subsection{$\chi_\lambda$} \label{app: compactsupport-chilambda}

Consider the adiabatic scaled switching family $\chi_\lambda = \chi(\tau/\lambda)$ \eqref{eq:phi-def-prep} in Proposition~\ref{prop:compact-prop}. 

$\chi$ is by assumption in~$C^1(\mathbb{R})$, absolutely integrable and square integrable, and not identically vanishing. 
These properties imply that $\chi_\lambda(\tau)$ is an ASSF, 
as was noted at the start of Section~\ref{subsubsec:adiab-revisited}. 

By continuity of~$\chi$, the bound \eqref{eq:thirdmoment-convergence}
implies that 
$\intinf \dd \tau \, |\tau^k \chi(\tau)| < \infty$ for $k=0,1,2,3$. 
A~dominated convergence argument then shows that $\widehat\chi$ is in~$C^3(\mathbb{R})$, 
$\widehat\chi^{(k)}(\omega) = (-\ii)^k \intinf \dd \tau \, \ee^{-\ii \omega\tau}\tau^k \chi(\tau)$ for $k=0,1,2,3$, and hence $\widehat\chi^{(k)}(0) = (-\ii)^k M_k[\chi]$ for $k=0,1,2,3$, where $M_k[\chi]$ are the moments of $\chi$ as defined in~\eqref{eq: moment def}. 

Since $M_0[\chi]=0$ by~\eqref{eq:cropped-zeroint}, it follows that 
$|\widehat\chi(\omega)| = O(\omega)$ as $\omega\to0$. 
The discussion around \eqref{eq:adiab-powerlaw-condition} then shows that when $\lambda(E)$ is the power-law~\eqref{eq: lambda = (S/E)^p} with $0 < \alpha < 2/3$, 
$\chi_\lambda$ satisfies the SFS property~\eqref{eq:SFS}. The small gap temperature is then positive and given by~\eqref{eq: non-zero temp}. 

This completes the proof of statements about $\chi_\lambda$ in Proposition~\ref{prop:compact-prop}.

\subsection{$\chi_{\lambda,\delta}$ is an ASSF} \label{app: compactsupport-part1}

Consider the switching family 
$\chi_{\lambda,\delta}$~\eqref{eq:phi-def}. We shall show that $\chi_{\lambda,\delta}$ is an ASSF (outcome 1 of
Proposition~\ref{prop:compact-prop}). 

Starting from \eqref{eq:phi-def}, we use the convolution theorem to write 
$\widehat{\chi}_{\lambda,\delta}$ in terms of $\widehat{\chi}$ and $\widehat{f}$ as 
\begin{align}
\label{eq:hatchi-conv}
\widehat{\chi}_{\lambda,\delta}(\omega)
= 
\frac{\lambda}{2\pi}
\intinf \dd \Omega \, 
\widehat{f}(\Omega)
\widehat{\chi}\bigl(\lambda\omega - \lambda^{1-\delta}\Omega \bigr) 
\,, 
\end{align}
where $\widehat{f}$ is smooth and falls off at infinity faster than any power, whereas $\widehat{\chi}$ is in $C^3(\mathbb{R})$ and square integrable. 

From \eqref{eq:hatchi-conv} we find 
\begin{align}
\frac{\widehat{\chi}_{\lambda,\delta}(u/\lambda)}{\lambda}
& = 
\frac{1}{2\pi}
\intinf \dd \Omega \, 
\widehat{f}(\Omega)
\widehat{\chi}\bigl(u - \lambda^{1-\delta}\Omega \bigr) 
\notag 
\\
& = 
\frac{1}{2\pi \lambda^{1-\delta}}
\intinf \dd \tilde \Omega \, 
\widehat{f}\bigl(\tilde\Omega / \lambda^{1-\delta}\bigr)
\widehat{\chi}\bigl(u - \tilde \Omega \bigr) , 
\label{eq:hatchi-conv2}
\end{align}
where the second equality comes by the change of variables $\Omega = \lambda^{\delta-1}\tilde\Omega$. 
From \eqref{eq:hatchi-conv2} we have the bound 
\begin{align}
\frac{|\widehat{\chi}_{\lambda,\delta}(u/\lambda)|}{\lambda}
\le \eta(u) , 
\end{align} 
where 
\begin{align}
\eta(u)
: = 
\sup_{\epsilon >0}
\frac{1}{2\pi}\left|
\intinf \dd \tilde \Omega \, 
 \epsilon^{-1} 
\widehat{f}\bigl(\tilde\Omega / \epsilon \bigr)  
  \widehat{\chi}\bigl(u - \tilde \Omega \bigr) \right| . 
\label{eq:hatchi-eta-def}
\end{align}
Because $\widehat f$ has rapid falloff and $\widehat\chi$ is square integrable, $\eta$ is square integrable
by Theorem III.2.(a) in~\cite{stein-sing-ints}, using the fact that $\widehat{\chi}$ has a square integrable maximal function by 
Theorem I.1.(c) in~\cite{stein-sing-ints}. This provides the $\eta$ for Proposition~\ref{prop:chilambda}. 

From \eqref{eq:xi} and \eqref{eq:hatchi-conv2} we have 
\begin{align}
\widehat\xi(u) 
&= 
\lim_{\lambda\to\infty}
\frac{\widehat{\chi}_{\lambda,\delta}(u/\lambda)}{\lambda}
\notag
\\
& = 
\frac{1}{2\pi}
\lim_{\lambda\to\infty}
\intinf \dd \Omega \, 
\widehat{f}(\Omega) 
\widehat{\chi}\bigl(u - \lambda^{1-\delta}\Omega \bigr)
\notag
\\
& = 
\widehat\chi(u) , 
\end{align}
taking the limit under the integral by dominated convergence, and using $f(0)=1$. Hence Proposition~\ref{prop:chilambda} applies with $\xi= \chi$. As $\chi$ is not identically vanishing, this completes the proof that $\chi_{\lambda,\delta}$ is an ASSF 
(outcome 1 of
Proposition~\ref{prop:compact-prop}).

\subsection{$\chi_{\lambda,\delta}$ satisfies SFS when $0<\alpha<2/3$ and $\delta> 3/2$} \label{app: compactsupport-part2}

What remains is to show that the switching family $\chi_{\lambda,\delta}$ \eqref{eq:phi-def} satisfies the SFS condition when $0<\alpha<2/3$ in the power-law $\lambda(E)$~\eqref{eq: lambda = (S/E)^p} 
and $\delta> 3/2$ 
(outcome 2 of
Proposition~\ref{prop:compact-prop}). 

Recall that the moments of~$\chi$, $M_k[\chi]$ \eqref{eq: moment def}, exists for $k=0,1,2,3$, by the assumption~\eqref{eq:thirdmoment-convergence}. Recall also that $M_0[\chi]=0$ by \eqref{eq: int chi = 0}. 

To begin, we assume $0<\alpha<1$ and $\delta>1$. 
We use \eqref{eq:phi-def} and the reality of $\chi$ and $f$ to write 
\begin{align} 
|\widehat{\chi}_{\lambda,\delta}(\omega)|^2 
&= \int_{{\mathbb{R}}^2} \dd u \, \dd v \,  \ee^{\ii\omega(u-v)} \chi(u/\lambda)\chi(v/\lambda)f(u/\lambda^\delta)f(v/\lambda^\delta)
\notag\\
&= \int_{{\mathbb{R}}^2} \dd u \, \dd v \,  \cos\bigl(\omega(u-v)\bigr)\chi(u/\lambda)\chi(v/\lambda)f(u/\lambda^\delta)f(v/\lambda^\delta)
\notag\\
&= \lambda^2
\int_{{\mathbb{R}}^2} \dd x \, \dd y \,  \cos\bigl(\omega \lambda (x-y)\bigr)\chi(x)\chi(y)
f(\lambda^{1-\delta}x) f(\lambda^{1-\delta}y)
\,, 
\label{eq: |varphi|^2}
\end{align}
where the second equality follows because the contribution from $\sin\bigl(\omega(u-v)\bigr)$ vanishes by antisymmetry under $(u,v)\to(v,u)$, 
and the third equality follows by the substitution $u=\lambda x$ and $v=\lambda y$.
The expression on the left-hand side of \eqref{eq:SFS} takes hence the form  
\begin{align} 
    \frac{1}{\lambda(E)} \int_{-|E|}^{|E|}\dd\omega \, {|\widehat{\chi}_{\lambda(E),\delta}(\omega)|}^2 &=  2S \lambda^{1-1/\alpha} 
\int_{{\mathbb{R}}^2} \dd x \, \dd y \,
    \sinc\bigl(S\lambda^{1-1/\alpha}(x-y) \bigr) \chi(x) \chi(y) f(\lambda^{1-\delta}x)f(\lambda^{1-\delta}y),
\label{eq: |varphi|^2 2}
\end{align}
interchanging the integrals by Fubini's theorem, performing the integral over $\omega$, and using \eqref{eq: lambda = (S/E)^p} to write $|E|$ in terms of~$\lambda$. 

To simplify the notation, we write
\begin{align} 
\epsilon = \frac{1}{\lambda} , 
\hspace{2ex}
a = \frac{1}{\alpha}-1,
\hspace{2ex}
b = \delta -1,
\label{eq: app D redefinition}
\end{align}
so that $\epsilon>0$, $a>0$, $b>0$, 
and the limit $\lambda\to\infty$ corresponds to $\epsilon \to 0$. 
From \eqref{eq: |varphi|^2 2} we then have 
\begin{subequations}
\label{eq:C-lhs-for_I}
\begin{align}
& \frac{1}{\lambda(E)} \int_{-|E|}^{|E|}\dd\omega \, {|\widehat{\chi}_{\lambda(E),\delta}(\omega)|}^2
= 2S \epsilon^{a} I_\epsilon
\,, \label{eq: varphi^2 and Iepsilon} \\
& I_\epsilon = \int_{{\mathbb{R}}^2} \dd x \, \dd y \,  \chi(x) \chi(y) \sinc\bigl(S \epsilon^a(x-y) \bigr)f(\epsilon^b x)f(\epsilon^b y)
\,. 
\label{eq: I epsilon}
\end{align}
\end{subequations}

Consider $I_\epsilon$~\eqref{eq: I epsilon}. Writing the integrand in terms of the functions $F(u) := f(u)-1$ and $G(u) := \sinc(u)-1$, we have 
\begin{align} 
I_\epsilon &= \int_{{\mathbb{R}}^2} \dd x \, \dd y \, \chi(x) \chi(y) \biggl[ G\bigl(S \epsilon^a (x-y) \bigr) + F(\epsilon^bx)F(\epsilon^b y) + G\bigl(S \epsilon^a (x-y) \bigr)\Bigl(F(\epsilon^bx)+F(\epsilon^b y)\Bigr)  
\notag \\
&\hspace{30ex}+G\bigl(S \epsilon^a (x-y) \bigr) F(\epsilon^bx)F(\epsilon^b y)\biggr], 
\label{eq: I epsilon 2}
\end{align}
where we have used
\begin{subequations}
\begin{align}
0 = \int_{{\mathbb{R}}^2} \dd x \, \dd y \, \chi(x) \chi(y) 
= \int_{{\mathbb{R}}^2} \dd x \, \dd y \, \chi(x) \chi(y) F(\epsilon^bx) = \int_{{\mathbb{R}}^2} \dd x \, \dd y \, \chi(x) \chi(y) F(\epsilon^by) \,,
\end{align}
\end{subequations}
which follows because $M_0[\chi]=0$. 

$F$ and $G$ are smooth bounded functions with the Maclaurin expansions 
$F(u) = f'(0) u + O(u^2)$ and $G(u) = -\tfrac16 u^2 + O(u^4)$. 
In the first and second term in \eqref{eq: I epsilon 2}, a Taylor expansion under the integral gives 
\begin{subequations}
\begin{align} 
\int_{{\mathbb{R}}^2} \dd x \, \dd y \, \chi(x) \chi(y) G\bigl(S \epsilon^a (x-y) \bigr) 
& = 
-\frac16 S^2 \epsilon^{2a}
\int_{{\mathbb{R}}^2} \dd x \, \dd y \, \chi(x) \chi(y) {(x-y)}^2
\ + o(\epsilon^{2a})
\notag
\\
& = 
\frac13 S^2 \epsilon^{2a}
{M_1[\chi]}^2
\ + o(\epsilon^{2a}) 
\,,
\label{eq:chi-chi-sinc}
\\
\intinf \dd x \, \chi(x) F(\epsilon^bx)
& = 
f'(0) \epsilon^b \intinf \dd x \, \chi(x)
 x 
\ + o(\epsilon^{b}) 
\notag
\\
& = 
f'(0) M_1[\chi]\epsilon^b
\ + o(\epsilon^{b}) 
\,,
\end{align}
\end{subequations}
where the error terms follow by a dominated convergence argument, using integrability of $x^2 |\chi(x)|$ and $|x \chi(x)|$, which follows from~\eqref{eq:thirdmoment-convergence}, 
and in \eqref{eq:chi-chi-sinc} we have used $M_0[\chi]=0$. 
The third and fourth term in \eqref{eq: I epsilon 2} are subdominant to the first term, by~\eqref{eq:thirdmoment-convergence}. 
Collecting, we have 
\begin{align} 
\label{eq:Ieps-final}
I_\epsilon
= 
{M_1[\chi]}^2 \! 
\left( \frac13 S^2 
 \epsilon^{2a}
+ 
\bigl(f'(0)\bigr)^2 \epsilon^{2b}
\right)
+ o(\epsilon^{2a}) + o(\epsilon^{2b}) 
\,. 
\end{align}
From \eqref{eq:C-lhs-for_I} and \eqref{eq:Ieps-final} we then have
\begin{align}
\label{eq:CondC-appvars}
\frac{1}{\lambda(E)} \int_{-|E|}^{|E|}\dd\omega \, {|\widehat{\chi}_{\lambda(E),\delta}(\omega)|}^2
= 
2 S {M_1[\chi]}^2 \! 
\left( \frac13 S^2 
 \epsilon^{3a}
+ 
\bigl(f'(0)\bigr)^2 \epsilon^{a + 2b}
\right)
+ o(\epsilon^{3a}) + o(\epsilon^{a + 2b}) 
\,. 
\end{align}

Now, the SFS condition \eqref{eq:SFS}
is equivalent to the condition that
\eqref{eq:CondC-appvars} is~$o(\epsilon^{1+a})$, 
using \eqref{eq: lambda = (S/E)^p} and~\eqref{eq: app D redefinition}. 
This happens when $a>1/2$ and $b>1/2$. 
In terms of $\alpha$ and~$\delta$, this happens when $0<\alpha<2/3$ and $\delta>3/2$. 

This completes the proof of outcome 2 of
Proposition~\ref{prop:compact-prop}.

\subsection{Variants of Proposition \ref{prop:compact-prop}} \label{app: compactsupportaux-variants}

If the falloff conditions of $\chi$ are strengthened, the small $\epsilon$ expansion of $I_\epsilon$ \eqref{eq: I epsilon} can be carried out to higher orders than shown in~\eqref{eq:Ieps-final}. If $\chi$ and $f$ are such that some of the coefficients of the low-order terms in this expansion vanish, the outcome is a variant of Proposition \ref{prop:compact-prop} where the domains of $\alpha$ or $\delta$ can be larger. We shall present here one such variant. 

We assume that both $\chi$ and $f$ are even, from which it follows that the coefficients of $\epsilon^{2a}$ and $\epsilon^{2b}$ in \eqref{eq:Ieps-final} vanish. We strengthen the falloff of $\chi$ from \eqref{eq:thirdmoment-convergence} to 
\begin{align}
\label{eq:fifthmoment-convergence}
\intinf \dd \tau \, |\tau^6 \chi(\tau)| < \infty
\,. 
\end{align}
The moments $M_k[\chi]$ \eqref{eq: moment def} are then defined for $k=0,1,\ldots,6$, and $M_0[\chi] = M_1[\chi] = M_3[\chi] = M_5[\chi] =0$. 
We now show that the error terms in \eqref{eq:Ieps-final} can be improved.

In \eqref{eq: I epsilon 2}, we use for $F$ and $G$ the Maclaurin expansions 
$F(u) = \tfrac12 f''(0) u^2 + O(u^4)$ and $G(u) = -\tfrac16 u^2 + \frac{1}{120}u^4 + O(u^6)$. For the first and second term in \eqref{eq: I epsilon 2}, we find 
\begin{subequations}
\begin{align} 
\int_{{\mathbb{R}}^2} \dd x \, \dd y \, \chi(x) \chi(y) G\bigl(S \epsilon^a (x-y) \bigr) 
& = 
\frac{1}{120} S^4 \epsilon^{4a}
\int_{{\mathbb{R}}^2} \dd x \, \dd y \, 
\chi(x) \chi(y) {(x-y)}^4
\ + o(\epsilon^{4a})
\notag
\\
& = 
\frac{1}{20} S^4 \epsilon^{4a}
{M_2[\chi]}^2
\ + o(\epsilon^{4a}) 
\,, 
\\
\intinf \dd x \, \chi(x) F(\epsilon^bx)
& = 
\frac12 f''(0) \epsilon^{2b} \intinf \dd x \, \chi(x)
 x^2 
\ + o(\epsilon^{2b}) 
\notag
\\
& = 
\frac12 f''(0) \epsilon^{2b} M_2[\chi]
\ + o(\epsilon^{2b}) 
\,. 
\end{align}
\end{subequations}
The third and fourth terms in \eqref{eq: I epsilon 2} are again subdominant to the first term. Collecting, we have 
\begin{align} 
\label{eq:Ieps-final-strong}
I_\epsilon
= 
{M_2[\chi]}^2 \! 
\left( \frac{1}{20} S^4
 \epsilon^{4a}
+ 
\frac14 
\bigl(f''(0)\bigr)^2 \epsilon^{4b}
\right)
+ o(\epsilon^{4a}) + o(\epsilon^{4b}) 
\,, 
\end{align}
from which \eqref{eq:C-lhs-for_I} gives 
\begin{align}
\label{eq:CondC-appvars-strong}
\frac{1}{\lambda(E)} \int_{-|E|}^{|E|}\dd\omega \, {|\widehat{\chi}_{\lambda(E),\delta}(\omega)|}^2
= 
2 S {M_2[\chi]}^2 \! 
\left( \frac{1}{20} S^4
 \epsilon^{5a}
+ 
\frac14 
\bigl(f''(0)\bigr)^2 \epsilon^{a+4b}
\right)
+ o(\epsilon^{5a}) + o(\epsilon^{a+4b}) 
\,. 
\end{align}
The SFS condition \eqref{eq:SFS} 
is equivalent to the statement that 
\eqref{eq:CondC-appvars-strong} is $o(\epsilon^{1+a})$, which happens when $a>1/4$ and $b>1/4$. 
In terms of $\alpha$ and~$\delta$, this happens when $0<\alpha<4/5$ and $\delta>5/4$. 

Collecting, we have shown that in this variant of proposition \ref{prop:compact-prop} the parameter ranges in outcome 2 are broadened to  $0<\alpha<4/5$ and $\delta>5/4$.

\bibliography{references}

\end{document}